\documentclass[a4paper,11pt]{article}
\usepackage{amsfonts,amsthm,amssymb}
\usepackage{dsfont}
\usepackage{tikz}
\newcommand\COMP{\hbox{C\kern -.58em {\raise .54ex \hbox{$\scriptscriptstyle |$}}
\kern-.55em {\raise .53ex \hbox{$\scriptscriptstyle |$}} }}
\newcommand\NN{\hbox{I\kern-.2em\hbox{N}}}
\newcommand\RR{\hbox{I\kern-.2em\hbox{R}}}
\newcommand\sRR{{\it \hbox{I\kern-.2em\hbox{R}}}}
\newcommand\QQ{\hbox{I\kern-.53em\hbox{Q}}}
\newcommand\PP{\hbox{I\kern-.53em\hbox{P}}}
\newcommand\EE{\hbox{I\kern-.53em\hbox{E}}}
\newcommand\ZZ{{{\rm Z}\kern-.28em{\rm Z}}}
\newcommand\be{\begin{equation}}
\newcommand\ee{\end{equation}}
\newtheorem{theorem}{Theorem}[section]

\newtheorem{proposition}[theorem]{Proposition}
\newtheorem{remark}[theorem]{Remark}

\newtheorem{example}[theorem]{Example}
\newtheorem{lemma}[theorem]{Lemma}

\newtheorem{definition}[theorem]{Definitions}
\newtheorem{corollary}[theorem]{Corollary}

\newcommand\beq{\begin{eqnarray}}
\newcommand\eeq{\end{eqnarray}}
\newcommand\bq{\begin{eqnarray*}}
\newcommand\eq{\end{eqnarray*}}

\def \Lbrack {[\![}
\def \Rbrack {]\!]}
\def \wid{\widetilde}

\def \cf{\cal F}
\def\mb{\mathbb}
\def\mbf{\mathbb F}
\def\mbg{\mathbb G}

\def \cf{{\cal F}}
\def \cg{{\cal G}}

\begin{document}
\title{Non-arbitrage  for Informational Discrete Time Market Models}

\author{Tahir
Choulli  \footnote{corresponding to: tchoulli@ualberta.ca,  Mathematical and Statistical Sciences Depart.,
University of Alberta, Edmonton, Canada } \  and     Jun Deng
\footnote{School of Banking and Finance, University of International Business and Economics, Beijing, China}
}

\maketitle
\begin{abstract}
This paper focuses on the stability of the  non-arbitrage condition in discrete time market models when some unknown information $\tau$ is partially/fully incorporated into the market. Our main conclusions are twofold. On the one hand, for a fixed market $S$, we prove that the non-arbitrage condition is preserved under a mild condition. On the other hand, we give the necessary and sufficient equivalent conditions on the unknown information $\tau$ to ensure the validity of the non-arbitrage condition for any market. Two concrete examples are presented to illustrate the importance of these conditions, where we calculate explicitly the arbitrage opportunities when they exist.
\end{abstract}





\section{Introduction}
In this paper, we pertain our attention to discrete time market models, where we consider a real-valued stochastic process $S=(S_n)_{0\leq n \leq N}$ that is indexed by the finite discrete time $\{0,1,...,N\}$. The process $S$ usually represents the  risky assets. \\

First, let us specify the definitions and notations. We suppose given a stochastic basis $(\Omega, \mathbb{A}, \mathbb{F}:= ({\cal F}_n)_{0\leq n \leq N}, \mathbb{P})$ and the process  $S=(S_n)_{0\leq n \leq N}$ is adapted to the filtration $\mathbb{F}$.   We say a process $X$ satisfies the non-arbitrage condition under the filtration $\mathbb{H}:= ({\cal H}_n)_{0\leq n \leq N}$ (hereafter, NA$(\mathbb{H})$) if
\begin{eqnarray}
 &&\mbox{for any predictable process  }  H:=(H_n)_{0\leq n \leq N}, \  (i.e. \  H_n \in {\cal H}_{n-1})\  \mbox{such that } \nonumber\\
 &&\sum_{1\leq n \leq N} H_n \Delta X_n \geq 0, \ \mbox{we have }  \sum_{1\leq n \leq N} H_n \Delta X_n \equiv 0, \ \mathbb{P}-a.s.
\end{eqnarray}
The process $H$ can be interpreted as the trading strategies that one holds dynamically through time. Loosely speaking, the non-arbitrage condition means there is no possibility that  one can   make profit out of nothing.   The equivalence between the non-arbitrage condition and equivalent martingale measure is essentially due to the work of Dalang, Morton and Willinger \cite{dalangmorton}, see also different approaches Schachermayer \cite{schachermayer92} and Rogers \cite{rogers94}.
\begin{theorem}[Dalang-Morton-Willinger]
 The process $X$ satisfies the non-arbitrage condition if and only  there exists an equivalent martingale measure. In this case, the equivalent martingale measure $\mathbb{Q}$ can be chosen with uniformly bounded density $d\mathbb{Q}/ d\mathbb{P}$.
\end{theorem}
 It   was baptized as The Fundamental Theorem of Asset Pricing. In this paper, we consider two economic agents with different information levels, one with the public available information $\mathbb{F}$ and an insider with some extra information beside $\mathbb{F}$.  Our goal  is to  study whether the insider with the extra information (characterized as a random time  $\tau$  in what follows) could make arbitrages.  The extra information $\tau$ could be the occurrence time of a default event, the knowledge that only insiders could get, and the last passage time of a process, etc.  For  continuous time settings, we refer to the recent works of Aksamit et al. \cite{cdm2013}, Acciaio et al. \cite{afk2014},  Choulli et al. \cite{cadj20141}, Coculescu et al. \cite{jeanblanc2012},  Fontana et al. \cite{fjs} and Song \cite{song2014}.\\

We begin with two examples that illustrate how the interplay of the random time $\tau$ and the market $S$ could affect  the non-arbitrage condition.
\begin{example}\label{example:1}
On the stochastic basis    $\left( {\Omega}, \mathbb{A}, \mathbb{F}:=({\cal F}_n)_{0\leq n \leq 2}, \mathbb{P} \right)$, we consider a two period discrete  model $S:=(S_n)_{0\leq n\leq 2}$, where $\Omega =\{\omega_1, \omega_2, \omega_3, \omega_4\}$ represents the uncertainties and the natural filtration  $\mathbb{F}:=({\cal F}_n)_{0\leq n \leq 2}$ is given by
\begin{eqnarray*}
{\cal F}_0 &=& \{\emptyset, \Omega\}, \ {\cal F}_1 = \{ \emptyset, \Omega, \{\omega_1, \omega_2\}, \{\omega_3, \omega_4\}\}, \mbox{and } {\cal F}_2 = \sigma(\{\emptyset, \Omega, \{\omega_1 \}, \{\omega_2 \}, \{\omega_3 \}, \{\omega_4 \} \}).
\end{eqnarray*}
Let $u$ and $d$ be two   constants such that $u>1$ and $0<d<1$.  Assume that
\begin{eqnarray*}
 S_1(\{\omega_1, \omega_2\}) &=& u S_0, \ \ S_1(\{\omega_3, \omega_4\}) = d S_0, \\
 S_2(\{\omega_1 \}) &=& u^2 S_0,\  S_2(\{\omega_2 \}) = u d S_0, \ S_2(\{\omega_3 \}) = u d S_0, S_2(\{\omega_4 \}) = d^2 S_0.
\end{eqnarray*}
The probability that the stock price will increase (or decrease) is $p$ (or $q=1-p$). We assume that the risk-free interest rate is zero and  $p u + (1-p)d=1$, i.e. $S$ is an $\mathbb{F}$-martingale under the physical probability $$\mathbb{P} = (\mathbb{P}(\omega_1), \mathbb{P}(\omega_2), \mathbb{P}(\omega_3),\mathbb{P}(\omega_4)) = (p^2, pq, pq, q^2).$$
The evolution of the stock price $S$ through time is illustrated as
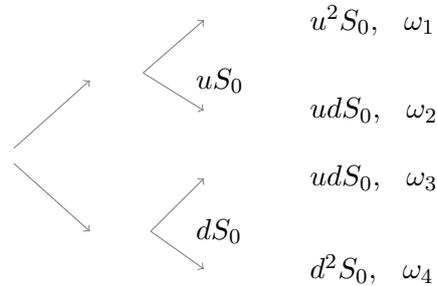
\begin{center}
\begin{tikzpicture}\label{stock1}
  node at (0,0) [below, black] {$S\rightarrow  S_0$};
   \draw [thin, gray, ->] (0,0.1) -- (1,1)      
        node [right, black] {$ u S_0 $};
         \draw [thin, gray, ->] (1.7,1.1) -- (2.5,1.8)      
        node [right, black] {$ u^2 S_0, \ \ \omega_1  $};
        \draw [thin, gray, ->] (1.7,1.1) -- (2.5,0.6)      
        node [right, black] {$ u d S_0, \ \ \omega_2  $};
        \draw [thin, gray, ->] (0,-0.1) -- (1,-1)      
        node [right, black] {$ d S_0 $};
        \draw [thin, gray, ->] (1.8,-1) -- (2.5,-0.3)      
        node [right, black] {$ u d S_0, \ \ \omega_3  $};
        \draw [thin, gray, ->] (1.8,-1) -- (2.5,-1.5)      
        node [right, black] {$ d^2 S_0, \ \ \omega_4  $};
\end{tikzpicture}
\end{center}
Consider the random time
\begin{equation}
 \tau = \left\{
 \begin{array}{cc}
  1, & \mbox{ on } \  \{\omega_3\}\\
  2, & \mbox{ otherwise}.
 \end{array}
 \right.
\end{equation}
Apparently, $\tau$ is not an $\mathbb{F}$-stopping time since $\{\tau =1\} \notin {\cal F}_1$. A straightforward calculation shows    the stopped market $S^\tau:=(S_{n\wedge \tau})_{0\leq n\leq 2}$ is given by
\begin{eqnarray*}
 S_0^\tau &=& S_0, \ \ S_1^\tau(\{\omega_1, \omega_2\}) = u S_0, \ \ S_1^\tau(\{\omega_3, \omega_4\}) = d S_0, \\
 S_2^\tau(\{\omega_1 \}) &=& u^2 S_0,\  S_2^\tau(\{\omega_2 \}) = u d S_0, \ S_2^\tau(\{\omega_3 \}) = d S_0, S_2^\tau(\{\omega_4 \}) = d^2 S_0.
\end{eqnarray*}
The evolution of the stock price $S^\tau$ through time is illustrated as
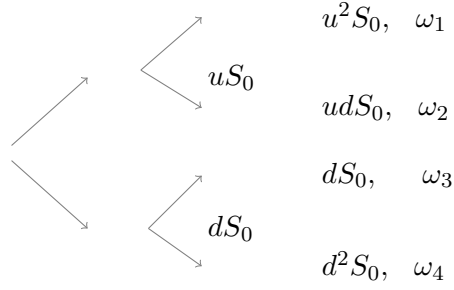
\begin{center}
\begin{tikzpicture}
  node at (0,0) [below, black] {$S^\tau \rightarrow S_0$};
   \draw [thin, gray, ->] (0,0.1) -- (1,1)      
        node [right, black] {$ u S_0 $};
         \draw [thin, gray, ->] (1.7,1.1) -- (2.5,1.8)      
        node [right, black] {$ u^2 S_0, \ \ \omega_1  $};
        \draw [thin, gray, ->] (1.7,1.1) -- (2.5,0.6)      
        node [right, black] {$ u d S_0, \ \ \omega_2  $};
        \draw [thin, gray, ->] (0,-0.1) -- (1,-1)      
        node [right, black] {$ d S_0 $};
        \draw [thin, gray, ->] (1.8,-1) -- (2.5,-0.3)      
        node [right, black] {$  d S_0, \ \ \  \ \omega_3 $};
        \draw [thin, gray, ->] (1.8,-1) -- (2.5,-1.5)      
        node [right, black] {$ d^2 S_0, \ \ \omega_4 $};
\end{tikzpicture}
\end{center}
Then, one could easily show that there exist  arbitrage opportunities in the market $S^\tau$. Indeed, a short selling on the scenarioes $\{\omega_3, \omega_4\}$ at time 1 would generate a sure profit.
\end{example}

\begin{example}\label{example:2}
We assume the same settings as Example \ref{example:1} and suppose  that
\begin{eqnarray*}
 S_1(\{\omega_1, \omega_2\}) &=& u S_0, \ \ S_1(\{\omega_3, \omega_4\}) = d S_0, \\
 S_2(\{\omega_1 \}) &=& u^2 S_0,\  S_2(\{\omega_2 \}) = u d S_0, \ S_2(\{\omega_3 \}) = u d S_0, S_2(\{\omega_4 \}) = d S_0.
\end{eqnarray*}
Set the physical probability $\mathbb{P}$ as
 $$\mathbb{P} = (\mathbb{P}(\omega_1), \mathbb{P}(\omega_2), \mathbb{P}(\omega_3),\mathbb{P}(\omega_4)) = \left(\frac{(1-d)^2}{(u-d)^2}, \frac{(u-1)(1-d)}{(u-d)^2}, \lambda \frac{u-1}{u-d}, (1-\lambda)\frac{u-1}{u-d}\right),$$
  where $0<\lambda<1$. Then, it is easy to see that  $S$ is an $\mathbb{F}$-martingale under $\mathbb{P}$ and is given by
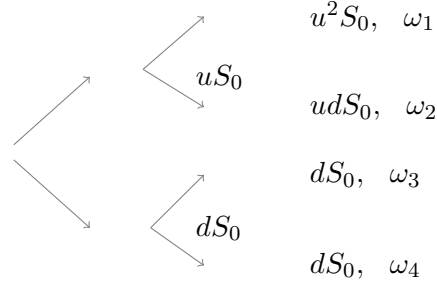
\begin{center}
\begin{tikzpicture}
  node at (0,0) [below, black] {$S\rightarrow S_0$};
   \draw [thin, gray, ->] (0,0.1) -- (1,1)      
        node [right, black] {$ u S_0 $};
         \draw [thin, gray, ->] (1.7,1.1) -- (2.5,1.8)      
        node [right, black] {$ u^2 S_0, \ \ \omega_1 $};
        \draw [thin, gray, ->] (1.7,1.1) -- (2.5,0.6)      
        node [right, black] {$ u d S_0, \ \ \omega_2 $};
        \draw [thin, gray, ->] (0,-0.1) -- (1,-1)      
        node [right, black] {$ d S_0 $};
        \draw [thin, gray, ->] (1.8,-1) -- (2.5,-0.3)      
        node [right, black] {$ d S_0, \ \ \omega_3  $};
        \draw [thin, gray, ->] (1.8,-1) -- (2.5,-1.5)      
        node [right, black] {$ d S_0, \ \ \omega_4  $};
\end{tikzpicture}
\end{center}
 Consider the random time
 \begin{equation}
   \tau_1 = \left\{
 \begin{array}{cc}
  1, & \mbox{ on } \  \{\omega_3\}\\
  2, & \mbox{ otherwise}.
 \end{array}
 \right.
 \end{equation}
One can easily show that $S^{\tau_1} = S$ since $\tau_1$ has no impact on $S$ on the scenarios $\{\omega_3,  \omega_4\}$. Therefore,  there is no arbitrage  opportunity  in $S^{\tau_1}$.
\end{example}

Motivated by these two examples, we are intending to find the necessary and sufficient conditions on $\tau$ or/and $S$ such that the market $S^\tau$ or $S-S^\tau$ still satisfies the non-arbitrage condition. We will come back to these two examples in the last section to explore why the non-arbitrage condition fails in  Example \ref{example:1} and holds in  Example \ref{example:2}.\\

The paper is organized as follows. Section \ref{sec:prelim} recalls some notations and definitions related to random time and progressive enlargement of filtration. In Section \ref{sec:beforetau}, we prove that the non-arbitrage condition is preserved for a fixed $\mathbb{F}$-martingale $S$ under some mild equivalent conditions on the stochastic interval $\Lbrack 0,\tau \Rbrack$; while in Section \ref{sec:aftertau} we aim at the non-arbitrage condition on the stochastic interval $\Rbrack \tau,+\infty \Lbrack$.  In the last section, we present two  examples to illustrate the importance of the conditions in  Section \ref{sec:beforetau} and  \ref{sec:aftertau}. Furthermore,  we construct explicitly the arbitrage opportunities when they exist.

\section{Preliminary}\label{sec:prelim}
On  a stochastic basis $(\Omega, \mathbb{A}, \mathbb{F}:= ({\cal F}_n)_{0\leq n \leq N}, \mathbb{P})$, we assume given an $\mathbb{F}$-adapted process  $S=(S_n)_{0\leq n \leq N}$ that represents the risky asset price and one risky-free asset that is assumed to be constant 1. In the market, we consider two economic agents, one with the public information $\mathbb{F}$  and an insider with the extra information $\tau$ and $\mathbb{F}$. These constitute the public information market $(\mathbb{F}, S)$ and the insider information market $(\mathbb{F}, S, \tau)$.  \\

 \noindent We start by recalling some notations and definitions related to the random time $\tau: \Omega \rightarrow \mathbb{Z}^+$ that would be fixed throughout this paper.  For any random time $\tau$, we associate the following two Az\'{e}ma supermartingales
\beq\label{eq:crucialZandZtilde}
Z_n:=P[\tau >n | \cf_n] \mbox{ and } \widetilde{Z}_n:= P[\tau \geq n | \cf_n],
\eeq
and the $\mbf$-stopping times
\beq\label{eq:crucialstoppingtime}
R_1:=\inf\{n\geq 0:  Z_n = 0\}, \ \ R_2:=\inf\{n\geq 1:  Z_{n-1} = 0\} \ \mbox{ and } R_3:=\inf\{n\geq 0:  \widetilde{Z}_n = 0\}.
\eeq
To incorporate the information from the random time $\tau$, we enlarge the filtration $\mbf$ by $\mb G= ({\cal G}_n)_{0\leq n\leq N}$
\beq
{\cal G}_n := {\cal F}_n\vee \sigma(\tau \leq n).
\eeq
In the literature,  $\mathbb{G}$ is called the progressive enlargement filtration that  is the smallest one that contains $\mathbb F$ and makes $\tau$  a stopping time. The insider information market is precisely characterized by $(\mathbb{G},S,\tau)$.
\begin{lemma}
  For any random time $\tau$ and the   stopping times in (\ref{eq:crucialstoppingtime}), the following   hold.\\
  \noindent $\rm (a)$ For all $n$,   $\{Z_{n-1} = 0\} \subset \{\widetilde{Z}_n = 0\} \subset \{Z_n = 0\}$.\\
  \noindent $\rm (b)$   $   R_1 \leq R_3\leq R_2=R_1 +1.$\\
 \noindent $\rm (c)$ On $\{n\leq \tau\}$,  $Z_{n-1}$ and $\wid{Z}_n$ are both positive. Consequently, $\tau \leq R_1$.
\end{lemma}
\begin{proof}
(a) Notice that
  \bq
  E\left[\widetilde{Z}_n \mathds{1}_{\{Z_{n-1} = 0\}}\right] = E\left[Z_{n-1} \mathds{1}_{\{Z_{n-1} = 0\}}\right] = 0.
  \eq
  Hence, $\{Z_{n-1}= 0 \} \subset \{\widetilde{Z}_n = 0\}$.  Due to $Z_n\leq \widetilde{Z}_n$, we have $\{\widetilde{Z}_n = 0\} \subset \{Z_n = 0\}$.\\

  (b) We observe that $\{R_2 = n\} = \{Z_{n-1} = 0\}\cap \left( \bigcup_{0\leq i\leq n-2} \{Z_i >0\} \right) = \{R_1 = n-1\}$. Therefore $R_2=R_1 +1$ and  is a predictable stopping time. The inequality $   R_1 \leq R_3 $ follows immediately from (a).\\

  (c) Notice that
  \begin{eqnarray*}
   E \left[\ \mathds{I}_{\{n\leq \tau\}} \mathds{I}_{\{Z_{n-1} = 0\}} \right] &=& E \left[ Z_{n-1} \mathds{I}_{\{Z_{n-1} = 0\}} \right]  = 0,  \ \mbox{and}\\
   E \left[\ \mathds{I}_{\{n\leq \tau\}} \mathds{I}_{\{\widetilde{Z}_{n} = 0\}} \right] &=& E \left[ \widetilde{Z}_{n}\mathds{I}_{\{\widetilde{Z}_{n} = 0\}} \right]  = 0.
  \end{eqnarray*}
  Therefore, $Z_{n-1}$ and $\wid{Z}_n$ are strictly positive on the set $\{n\leq \tau\}$ and $\tau\leq R_1$.
 \end{proof}

 \begin{remark}
  It was proved in Dellacherie and Meyer \cite{dm2}  these three sets $\{Z_{-} = 0\}$, $\{Z= 0\}$ and $\widetilde{Z} = 0$  have the same d\'ebut   in continuous time setting that discrete time does not share.
 \end{remark}

\begin{lemma}\label{lem:decompositionofZ}
  The Az\'{e}ma supermartingale $(Z_n)_{0\leq n\leq N}$ has the following decomposition.
  \beq
  Z_n = m_n - A_n,\ \  m_n := P[\tau >n | \cf_n] + \sum_{0\leq k\leq n}P[\tau =k | \cf_k], \ A_n:=\sum_{0\leq k\leq n}P[\tau =k | \cf_k],
  \eeq
  where $(m_n)_{0\leq n\leq N}$ is an $\mbf$-martingale and $(A_n)_{0\leq n\leq N}$ is an $\mbf$-adapted increasing process.
\end{lemma}
\begin{proof}
  It is enough to prove $(m_n)_{0\leq n \leq N}$ is an $\mbf$-martingale. To this end, we calculate that
  \beq
  E\Big[m_{n+1} \Big| \cf_n \Big] &=& P\Big[\tau >{n+1} \Big| \cf_n\Big] + \sum_{0\leq k\leq n+1}E\Big[ P[\tau =k | \cf_k]\Big| \cf_n\Big] \nonumber \\
  &=& P\Big[\tau >{n+1} \Big| \cf_n\Big]+ \sum_{0\leq k\leq n}P\Big[\tau =k | \cf_k\Big] +P\Big[\tau = n+1 \Big| \cf_n\Big]  \nonumber \\
  &=& P[\tau >n | \cf_n] + \sum_{0\leq k\leq n}P[\tau =k | \cf_k] = m_n. \nonumber
  \eeq
  This ends the proof of the lemma.
 \end{proof}
 \begin{remark}
 In general, the decomposition $Z = m - A$ is not the Doob-Meyer decomposition.
 \end{remark}

The following lemma describes the connection between conditional expectations under $\mbf$ and $\mbg$. For its proof, we consult Jeulin \cite{Jeu}.
\begin{lemma}\label{lem:OPunderGF}
Let $Y$ be an integrable and $\mb A$-measurable  random variable. Then, the following   hold.\\
$\rm (a)$ On the set $\{n<\tau\}$, the conditional expectation of $Y$ under ${\cal G}_n$ is given by
\beq
    E\left[Y|{\cal G}_n\right]\mathds{1}_{\{\tau > n\}} = E\left[Y \mathds{1}_{\{\tau > n\}} |{\cal F}_n\right] \frac{1}{Z_n}\mathds{1}_{\{\tau > n\}}.
    \eeq
$\rm (b)$ On the set $\{n\leq \tau\}$, the conditional expectation of $Y$ under ${\cal G}_{n-1}$ is given by
\beq
    E\left[Y|{\cal G}_{n-1}\right]\mathds{1}_{\{\tau \geq n\}} = E\left[Y \mathds{1}_{\{\tau \geq n\}} |{\cal F}_{n-1}\right] \frac{1}{Z_{n-1}}\mathds{1}_{\{\tau \geq n\}}.
    \eeq
Moreover, if $Y$ is ${\cal F}_n$-measurable, we have
\beq\label{eq:OPunderGF}
E\left[Y|{\cal G}_{n-1}\right]\mathds{1}_{\{\tau \geq n\}} &=& E\left[Y \widetilde{Z}_n |{\cal F}_{n-1}\right] \frac{1}{Z_{n-1}}\mathds{1}_{\{\tau \geq n\}}.
\eeq
\end{lemma}

\section{Non-arbitrage on $\Lbrack 0, \tau \Rbrack$}\label{sec:beforetau}
In this section, we will prove that the non-arbitrage condition is preserved under one mild condition when the market is stopped at  random horizon $\tau$. Furthermore, we gave the necessary and sufficient conditions (on $\tau$ or the stopping times in (\ref{eq:crucialstoppingtime})) to guarantee the stability of the non-arbitrage condition for  any  market $S^\tau$. \\

\noindent The following theorem characterizes the relationship between $\mathbb F$-martingales and $\mathbb G$-martingales. For the continuous time settings, we consult Jeulin \cite{Jeu}.
\begin{theorem}\label{theo:GFmartau}
  Let $M$ be an $\mbf$-martingale and $\tau$ be a random time. Then the following process
  \beq\label{eq:gmartingalebeforetau}
  \widehat{M}^{(b)}_n := M_{n\wedge \tau} - \sum_{1\leq k\leq n} \frac{1}{Z_{k-1}}\mathds{1}_{\{\tau \geq k\}} E\left[\Delta M_k  \widetilde{Z}_k | \cf_{k-1}\right],
  \eeq
  is a $\mbg$-martingale.
\end{theorem}
\begin{proof}
  Although it can be derived from Jeulin \cite{Jeu}, we opt to give a direct proof here.  To this end, we calculate
  \bq
  E\Big[\Delta \widehat{M}^{(b)}_{n} \Big| \cg_n\Big] &=& E\Big[ \Delta M_{n} I_{\{n\leq \tau\}} -  \frac{1}{Z_{n-1}}\mathds{1}_{\{n\leq \tau\}} E\left[\Delta M_n  \widetilde{Z}_n | \cf_{n-1}\right] \Big| \cg_{n-1}\Big]\nonumber \\
  &=&  E\left[\Delta M_n \widetilde{Z}_n |{\cal F}_{n-1}\right]  \frac{1}{Z_{n-1}}\mathds{1}_{\{n\leq \tau\}}-  \frac{1}{Z_{n-1}}\mathds{1}_{\{n\leq \tau\}} E\left[\Delta M_n  \widetilde{Z}_n | \cf_{n-1}\right] =0,
  \eq
  where in the above  second equality we use the fact (due to Lemma \ref{lem:OPunderGF}-(b)) that
  \beq
  E\Big[  \Delta M_n \mathds{1}_{\{\tau \geq n\}} \big\vert \cg_{n-1}\Big] = \frac{1}{Z_{n-1}}\mathds{1}_{\{\tau \geq n\}}E\left[\Delta M_n \widetilde{Z}_{n} \big\vert \cf_{n-1}\right]. \nonumber
  \eeq
  This ends the proof of theorem.
 \end{proof}
In the following proposition, we construct a $\mathbb G$-martingale that would    serve as the martingale density for a class of $\mathbb{G}$-semi-martingales.
\begin{proposition}\label{prop:density}
The following process
 \beq\label{eq:mardensity}
  \widehat{N}^{(b)}_n:= -\sum_{1\leq k\leq n} \mathds{1}_{\{\tau \geq k\}}E[\mathds{1}_{\{\widetilde{Z}_{k} >0\}} | {\cal F}_{k-1}] + \sum_{1\leq k\leq n} \frac{Z_{k-1}}{\widetilde{Z}_k}\mathds{1}_{\{\tau \geq k\}}
  \eeq
  is a $\mbg$-martingale such that $1 + \Delta \widehat{N}^{(b)}_n >0$ for all $n\geq 1$.
\end{proposition}

\begin{proof}
  First, we prove that $\widehat{N}^{(b)}$ is a $\mb G$-martingale. To this end, by using Lemma \ref{lem:OPunderGF}-(b), we calculate
  \beq
  E\left[\Delta \widehat{N}^{(b)}_{n} \Big\vert {\cal G}_n\right] &=& E\left[-  \mathds{1}_{\{\tau \geq n\}}E\left[\mathds{1}_{\{\widetilde{Z}_{n} >0\}} \Big\vert {\cal F}_{n-1}\right] + \frac{Z_{n-1}}{\widetilde{Z}_n}\mathds{1}_{\{\tau \geq n\}} \Big| {\cal G}_n\right] \nonumber \\
  &=& -  \mathds{1}_{\{\tau \geq n\}}E\left[\mathds{1}_{\{\widetilde{Z}_{n} >0\}} \Big\vert {\cal F}_{n-1}\right] + E\left[\frac{Z_{n-1}}{\widetilde{Z}_{n}}\mathds{1}_{\{\tau \geq n\}} \Big| {\cal G}_n\right] \nonumber \\
   &=& -  \mathds{1}_{\{\tau \geq n\}}E\left[\mathds{1}_{\{\widetilde{Z}_{n} >0\}} \Big\vert {\cal F}_{n-1}\right] + \mathds{1}_{\{\tau \geq n\}}E\left[\mathds{1}_{\{\widetilde{Z}_{n} >0\}} \Big\vert {\cal F}_{n-1}\right] =0. \nonumber
  \eeq

  \noindent Secondly, we check the integrability of $\widehat{N}^{(b)}$. Indeed,
  \bq
  E[|\widehat{N}^{(b)}_n|] &\leq& n + \sum_{1\leq k\leq n} E\left[\frac{Z_{k-1}}{\widetilde{Z}_k} \mathds{1}_{\{\tau \geq k\}}\right] = n + \sum_{1\leq k\leq n} E\left[Z_{k-1} \mathds{1}_{\{\widetilde{Z}_k>0\}}\right] \leq 2n.
  \eq
  \noindent Finally, we show that $1 + \Delta \widehat{N}^{(b)}_n >0$. Indeed
  \bq
  1 + \Delta \widehat{N}^{(b)}_n = 1  -\mathds{1}_{\{\tau \geq n\}}E[\mathds{1}_{\{\widetilde{Z}_{n} >0\}} | {\cal F}_{n-1}] + \frac{Z_{n-1}}{\widetilde{Z}_n} \mathds{1}_{\{\tau \geq n\}} \geq  \mathds{1}_{\{\tau < n\}} + \frac{Z_{n-1}}{\widetilde{Z}_n} \mathds{1}_{\{\tau \geq n\}} >0.
  \eq
  This completes the proof of the proposition.
 \end{proof}
\begin{remark}
  In fact, $ \sum_{1\leq k\leq n} \mathds{1}_{\{\tau \geq k\}}E[\mathds{1}_{\{\widetilde{Z}_{k} >0\}} | {\cal F}_{k-1}] $ is the $\mbg$-compensator of the $\mbg$-adapted increasing process $ \sum_{1\leq k\leq n} {Z_{k-1}}/{\widetilde{Z}_k}\mathds{1}_{\{\tau \geq k\}}$.
\end{remark}

\begin{lemma}
  The stochastic exponential  ${\cal E}(N)$ of a local martingale $N$ is the  form of
  \beq
  {\cal E}(N)_n = \prod_{1\leq k\leq n} (1+\Delta N_k).
  \eeq
\end{lemma}
\begin{proof}
  It is straightforward from the calculation of the stochastic exponential.
 \end{proof}

Now, we are ready to state our first main theorem for this section.
\begin{theorem}\label{theo:maintheotau}
  Consider any random time $\tau$  and the $\mathbb F$-martingale $S$. Denote the probability measure $\mathbb{Q} \sim \mathbb{P}$ with density $D_n:={\cal E}(Y)_n$ where
  \beq\label{eq:crucialmeasureQbtau}
  \Delta Y_n := \widetilde{Z}_{n}\mathds{1}_{\{Z_{n-1}>0\}} E\left[\mathds{1}_{\{\widetilde{Z}_{n} =0\}} | {\cal F}_{n-1}\right] - Z_{n-1} \mathds{1}_{\{\widetilde{Z}_{n}=0 \}}, \ Y_0=0.
  \eeq
  Then the following are equivalent:\\
  $\rm (a)$ $S$ is an $(\mathbb F,\mathbb{Q})$-martingale;\\
  $\rm (b)$ $S$ is orthogonal to $D$ and $Y$;\\
  $\rm (c)$ ${\cal E}(\widehat{N}^{(b)})_n S_{n\wedge \tau}$ is a $\mathbb G$-martingale, where $\widehat{N}^{(b)}$ is given by (\ref{eq:mardensity}).\\
  As a consequence, all the above three equivalent conditions imply that:\\
  $\rm (d)$ $S^\tau$ satisfies NA$(\mathbb G, \mathbb{P})$ and NA$(\mathbb G, \mathbb{Q})$.
\end{theorem}

\begin{proof}
First, we remark that the probability measure $\mathbb{Q}$ is well defined and  equivalent to $\mathbb{P}$. Indeed, it is easy to check that $(Y_n)$ is an $\mathbb F$-martingale and $$
1 + \Delta Y_n = \widetilde{Z}_n\mathds{1}_{\{Z_{n-1}>0\}} E\left[\mathds{1}_{\{\widetilde{Z}_n=0\}} | {\cal F}_{n-1}\right] + \mathds{1}_{\{\widetilde{Z}_n>0\}} + (1-Z_{n-1})\mathds{1}_{\{\widetilde{Z}_n=0\}} > 0,
$$
where we used the fact that on the set $\{\widetilde{Z}_n>0\}$, $1 + \Delta Y_n \geq 1$ and the inclusion  $\{\widetilde{Z}_n=0\} \subset \{Z_{n-1} <1\}$, since $\{Z_{n-1} =1\}\subset \{\widetilde{Z}_n=1\} $.  Therefore, $D$ is a strictly positive martingale.\\

The equivalence between (a) and (b) is obvious. In the following, we are focusing on the proof of the equivalence between (a) and (c).  Recall that
  \beq
  \widehat{N}^{(b)}_n = -\sum_{1\leq k\leq n} \mathds{1}_{\{\tau \geq k\}}E[\mathds{1}_{\{\widetilde{Z}_{k} >0\}} | {\cal F}_{k-1}] + \sum_{1\leq k\leq n} \frac{Z_{k-1}}{\widetilde{Z}_k}\mathds{1}_{\{\tau \geq k\}}.
  \eeq
     Due to Lemma \ref{lem:OPunderGF}, we deduce that
  \beq
  E\Big[\frac{\Delta S_k }{\widetilde{Z}_k}  \mathds{1}_{\{\tau \geq k\}} | {\cal G}_{k-1}\Big] &=&  \frac{\mathds{1}_{\{\tau \geq k\}}}{Z_{k-1}}E\Big[\Delta S_k \mathds{1}_{\{\widetilde{Z}_k >0\}}  | {\cal F}_{k-1}\Big],\nonumber \\
  E\Big[ \Delta S_k  \mathds{1}_{\{\tau \geq k\}} | {\cal G}_{k-1}\Big] &=& \frac{\mathds{1}_{\{\tau \geq k\}}}{Z_{k-1}}E\Big[ \Delta S_k \widetilde{Z}_k | {\cal F}_{k-1}\Big].
  \eeq
  To this end, we calculate that
  \bq
  &&E\left[{\cal E}(\widehat{N}^{(b)})_{n+1} S_{(n+1)\wedge \tau} \Big| {\cal G}_{n+1} \right] \\
  && ={\cal E}(\widehat{N}^{(b)})_n E\Big[ (1+\Delta \widehat{N}^{(b)}_{n+1}) S_{(n+1)\wedge \tau} \Big| {\cal G}_{n} \Big] \\
  && = {\cal E}(\widehat{N}^{(b)})_n E\Big[ S_{n\wedge \tau} + \Delta S_{n+1}\mathds{1}_{\{n+1\leq \tau\}} + \Delta\widehat{N}^{(b)}_{n+1} S_{n\wedge \tau} + \Delta S_{n+1} \Delta \widehat{N}^{(b)}_{n+1} \mathds{1}_{\{n+1 \leq \tau\}}\Big| {\cal G}_{n} \Big] \\
  &&={\cal E}(\widehat{N}^{(b)})_n \left\{ S_{n\wedge \tau} + E\Big[ \Delta S_{n+1}\widetilde{Z}_{n+1} E\left[\mathds{1}_{\{\widetilde{Z}_{n+1}=0 \}}| {\cal F}_{n}\right]\Big| {\cal F}_n \Big]\frac{\mathds{1}_{\{n+1\leq \tau\}}}{Z_n}  \right\}\\
  &&\ \  - \ {\cal E}(\widehat{N}^{(b)})_n \left\{ \mathds{1}_{\{n+1\leq \tau\}} E\Big[ \Delta S_{n+1} \mathds{1}_{\{\widetilde{Z}_{n+1}=0\}} \Big| {\cal F}_n \Big]\right\}\\
  &&= {\cal E}(\widehat{N}^{(b)})_n  S_{n\wedge \tau} + {\cal E}(\widehat{N}^{(b)})_n \left\{E\left[\Delta S_{n+1}\left\{\widetilde{Z}_{n+1} E\left[\mathds{1}_{\{\widetilde{Z}_{n+1} =0\}} | {\cal F}_n\right] - Z_n \mathds{1}_{\{\widetilde{Z}_{n+1} =0\}} \right\} \Big| {\cal F}_n\right]\right\}\frac{\mathds{1}_{\{n+1\leq \tau\}}}{Z_n}\\
&& = {\cal E}(\widehat{N}^{(b)})_n  S_{n\wedge \tau} + {\cal E}(\widehat{N}^{(b)})_n E^\mathbb{Q}\left[\Delta S_{n+1} \Big| {\cal F}_n\right]\frac{\mathds{1}_{\{n+1\leq \tau\}}}{Z_n} .
\eq
Therefore, (a) implies (c). Conversely, if (c) holds, we have $$E^\mathbb{Q}\left[\Delta S_{n+1} | {\cal F}_n\right]\frac{\mathds{1}_{\{n+1\leq \tau\}}}{Z_n} = 0, \ \mbox{and }  \ E^\mathbb{Q}\left[\Delta S_{n+1} | {\cal F}_n\right]\mathds{1}_{\{Z_n>0\}} = 0.$$ Notice that $$E^\mathbb{Q}\left[\Delta S_{n+1} | {\cal F}_n\right]\mathds{1}_{\{Z_n=0\}} = 0, \mbox{ for all $n$ }.$$ Thus, we conclude that $E^\mathbb{Q}\left[\Delta S_{n+1} | {\cal F}_n\right] = 0$, for all $n$. This completes the proof of the theorem.
 \end{proof}

\begin{remark}
  We observe from Theorem  \ref{theo:maintheotau} that even though $Y$ is an $\mathbb F$-martingale, the stopped process $Y_{n\wedge \tau} =  \sum_{k\leq n}\widetilde{Z}_{k}E\left[\mathds{1}_{\{\widetilde{Z}_{k} =0\}} | {\cal F}_{k-1}\right]\mathds{1}_{\{k\leq \tau\}}$ does not satisfy NA$(\mathbb G)$ since it is a $\mathbb G$-increasing process. This also sheds some light on the importance of the conditions in   Theorem  \ref{theo:maintheotau}.
\end{remark}

\begin{remark}
  It is worthy to notice that, in general, for an $\mathbb F$-martingale $M$, if $M^\tau$ satisfies NA$(\mathbb G)$, we can not conclude $M$ is orthogonal to $Y$. To wit, let the projection of $Y$ to $m$ as
  \begin{eqnarray*}
  \Delta Y_n = H_n \Delta m_n + \Delta \overline{m}_n,
   \end{eqnarray*}
  where $H_n\in {\cal F}_{n-1}$  and $\overline{m}$ is an $\mathbb F$-martingale, orthogonal to $m$. If $Y$ is not null, $\overline{m}$ is not identical zero. By Theorem \ref{theo:GFmartau}, it is easy to see that $\overline{m}^\tau$ is a $\mathbb G$-martingale. However, $\overline{m}$ can not be orthogonal to $Y$ unless $Y$ is null.  
\end{remark}

\begin{corollary}\label{theo:cruialthereomforMuptotau}
  Let $M$ be an $\mbf$-martingale. If for all $n$,
  \beq\label{eq:crucialassumptionbeforetau}
  \{\widetilde{Z}_n =0\} = \{Z_{n-1} =0\}.
  \eeq
  Then the following properties hold:\\
  $\rm (a)$ $(M_{n\wedge \tau})_{n\geq 1}$ satisfies NA($\mbg$);\\
  $\rm (b) $ $\left({\cal E}(\widehat{N}^{(b)})_n M_{n\wedge \tau}\right)_{n\geq 1}$ is a $\mbg$-martingale, where $\widehat{N}^{(b)}$ is given by (\ref{eq:mardensity}) in Proposition \ref{prop:density};\\
  $\rm (c)$ The probability measure $\mathbb{Q}$ given in (\ref{eq:crucialmeasureQbtau}) coincides with $\mathbb{P}$, .\\
  Particularly, the above three properties hold when $Z_n>0$ for all $n\geq 0$.
\end{corollary}

\noindent Below, we state our second main theorem of this section, where we give the necessary and sufficient conditions that imposed on the random time $\tau$ (or the stopping times in (\ref{eq:crucialstoppingtime})) to guarantee the stopped process $M^\tau$ satisfies NA($\mbg$) for any $\mbf$-martingale $M$.
\begin{theorem}\label{theo:main2before}
Consider a random time $\tau$ and the associated stopping times defined in (\ref{eq:crucialstoppingtime}). Then the following are equivalent:\\
$\rm (a)$ For any $\mbf$-martingale $M$, the stopped process $M^\tau$ satisfies NA$(\mbg)$.\\
$\rm (b)$   $\{\widetilde{Z}_n = 0\} =\{Z_{n-1} = 0\}$ for all $n$. \\
$\rm (c)$ $  R_1 + 1 = R_2 =R_3$.\\
$\rm (d)$ $R_3$ is an $\mathbb{F}$-predictable stopping time.\\
$\rm (e)$ The probability $\mathbb{Q}$, defined via (\ref{eq:crucialmeasureQbtau}), coincides with $\mathbb{P}$.
\end{theorem}
\begin{proof}
The proof of the theorem would be achieved after four steps. In the first step, we prove (b)$\Leftrightarrow$(c). The second step  focuses on (c)$\Leftrightarrow$(d). The third step deals with (b)$\Leftrightarrow$(e). In the last step, we prove (a) $\Leftrightarrow$ (b).  \\

\noindent  Step 1: The equivalence between (b) and (c) is obvious. Indeed, if (b) holds, it is trivial that $R_2 = R_3$. Conversely, if (c) holds, we derive that $$E\left( Z_{n-1} I_{\{\widetilde{Z}_n = 0\}} \right) = E\left( Z_{n-1} I_{\{\widetilde{Z}_n = 0\}} I_{\{n \geq R_3\}} \right)=E\left( Z_{n-1} I_{\{\widetilde{Z}_n = 0\}} I_{\{n \geq R_2\}} \right) = 0.$$
  Hence, we conclude that $\{\widetilde{Z}_n = 0\} \subset \{Z_{n-1} = 0\}$ for all $n$. \\

  \noindent Step 2: We prove (c)$\Leftrightarrow$(d). If (c) holds, it is easy to see that  $R_3$ is an $\mathbb{F}$-predictable stopping time due to $\{R_3 = n\} = \{R_1 = n-1\} \in {\cal F}_{n-1}$. Conversely, by the predictability of $R_3$, we have  $0=E[\widetilde{Z}_{R_3}] = E[Z_{R_3 - 1}]$; hence  $Z_{R_3 - 1} = 0$ and $R_3 = R_2$.\\

   \noindent Step 3:  We prove (b)$\Leftrightarrow$(e). If (b) holds, apparently, $Y =0$ and $\mathbb{Q} =\mathbb{P}$. Conversely, if (e) holds, $\Delta Y_n  = 0$ for all $n$. Hence, $\widetilde{Z}_{n}\mathds{1}_{\{Z_{n-1}>0\}} E\left[\mathds{1}_{\{\widetilde{Z}_{n} =0\}} | {\cal F}_{n-1}\right] = Z_{n-1} \mathds{1}_{\{\widetilde{Z}_{n}=0 \}}=0$ and $\{\widetilde{Z}_{n}=0 \} = \{{Z}_{n-1}=0 \}$ for all $n$. \\

 \noindent  Step 4:  In this step, we focus on the proof of the equivalence between (a) and (b).\\
  (a)$\Rightarrow$(b). Suppose for any $\mbf$-martingale $M$, the stopped process $M^\tau$ satisfies NA($\mbg$). Consider
  \beq
  V_n := \mathds{1}_{\{R_3>n\}} \ \  \mbox{ and } \ \ \widetilde{V}_n :=\sum_{1\leq k\leq n} \left\{E[V_k | {\cal F}_{k-1}] - V_{k-1}\right\}.
  \eeq
  It is easy to see that $M_n := V_n - \widetilde{V}_n$ is an $\mbf$-martingale. Therefore $M_{n\wedge \tau} = 1 - \widetilde{V}_{n\wedge \tau}$ satisfies NA($\mbg$). Then there exists an equivalent probability $\mathbb{Q}_1\sim \mathbb{P}$ such that $\widetilde{V}_{n\wedge \tau}$ is a $(\mbg,\mathbb{Q}_1)$-martingale. Therefore $\widetilde{V}_{n\wedge \tau} \equiv 0.$ Hence, we have
  \bq
  0&=& E\left[\widetilde{V}_{n\wedge \tau}\right] = E\left[\sum_{1\leq k\leq n} Z_{k-1}\Delta \widetilde{V}_k\right] \\
  &=&\sum_{1\leq k\leq n} E\left[ Z_{k-1}\left(E\left[\mathds{1}_{\{R_3>k\}} | {\cal F}_{k-1}\right] - \mathds{1}_{\{R_3 >k-1\}}\right)\right]\\
  &=&\sum_{1\leq k\leq n} E\left[ Z_{k-1} \mathds{1}_{\{R_3>k\}} \right] - E\left[ Z_{k-1} \mathds{1}_{\{R_3 >k-1\}}\right]\\
  &=& - \sum_{1\leq k\leq n} E\left[ Z_{k-1} \mathds{1}_{\{R_3=k\}} \right] = -\sum_{1\leq k\leq n} E\left[ Z_{k-1} \mathds{1}_{\{\widetilde{Z}_k = 0\}}\prod_{1\leq i\leq k}\mathds{1}_{\{\widetilde{Z}_{i-1} > 0\}} \right]\\
  &=& -\sum_{1\leq k\leq n} E\left[ Z_{k-1} \mathds{1}_{\{Z_{k-1} >0\}}\mathds{1}_{\{\widetilde{Z}_k = 0\}}\prod_{1\leq i\leq k}\mathds{1}_{\{\widetilde{Z}_{i-1} > 0\}} \right]\\
  &=& -\sum_{1\leq k\leq n} E\left[ Z_{k-1} \mathds{1}_{\{Z_{k-1} >0\}}\mathds{1}_{\{\widetilde{Z}_k = 0\}} \right]=-\sum_{1\leq k\leq n} E\left[ Z_{k-1}\mathds{1}_{\{\widetilde{Z}_k = 0\}} \right],
  \eq
  where we used the fact that $\{Z_{k} >0\} \subset \{\widetilde{Z}_k>0\} \subset \{\widetilde{Z}_{k-1}>0\} $.   Therefore, for all $n$, $\{\widetilde{Z}_n = 0\} \subset \{Z_{n-1} = 0\}$ and $R_3 \geq R_2$.  \\

\noindent  (b)$\Rightarrow$(a) It follows immediately from Theorem \ref{theo:maintheotau} or Corollary \ref{theo:cruialthereomforMuptotau}. This ends  the proof of the theorem.
 \end{proof}
 An interesting  corollary for two period model  is
 \begin{corollary}
  Consider a two period model $(\Omega, {\cal A}={\cal F}_2, \mathbb{F}:=({\cal F}_n)_{n=0,1,2}, \mathbb{P})$ with an $\cal A$-measurable positive random time $\tau$. For any $\mbf$-martingale $M$, the   process $M^\tau$ satisfies NA$(\mbg)$ if and only if $\tau$ is an $\mathbb{F}$-stopping time.
 \end{corollary}

 \begin{proof}
 If $\tau$ is an $\mathbb{F}$-stopping time, it is trivial that  $M^\tau$ satisfies NA$(\mbg)$ for any $\mbf$-martingale $M$. Conversely, for the random time $\tau$, denote $\Omega_2:=\{\tau = 2\}$, $\Omega_1:=\{\tau = 1\}$ and $\Omega_1 \cup \Omega_2 = \Omega$. By the definitions of $Z$ and $\widetilde{Z}$, we derive that
  \begin{eqnarray*}
   \widetilde{Z}_0= 1, \widetilde{Z}_1 = 1, \widetilde{Z}_2 = I_{\Omega_2}, \ \mbox{and} \ Z_0 = 1, Z_1 = \mathbb{P}(\Omega_2 | {\cal F}_1), Z_2 = 0.
  \end{eqnarray*}
If for  any $\mbf$-martingale $M$, the stopped process $M^\tau$ satisfies NA$(\mbg)$, by Theorem \ref{theo:main2before}, we know that $\{\widetilde{Z}_2 = 0\} = \Omega_1 = \{Z_{1} = 0\}\in {\cal F}_1$ and $\tau$ is an $\mathbb{F}$-stopping time.
 \end{proof}

 \subsection{Reverse Problem: before $\tau$} 
 The previous section studied  what we can conclude for arbitrage opportunities  from the standpoint view of the insider. In this section, we will investigate the equivalence or consequence on the market $S$ if we know that the insider can not make arbitrage opportunities in the market $(\mathbb{G},S^\tau)$.\\
  
We start with two simple lemmas and one proposition before proving Theorem \ref{theo:reversemainbefore} below.

 \begin{lemma}\label{DIStheincluseion111}
  The following hold.
  \begin{eqnarray}
   \{n\leq \tau\} \  \subset \ \{\widetilde{Z}_n>0\} \  \subset \   \{Z_{n-1} >0\}  = \    \Gamma(n):=\left\{P\left(\widetilde Z_n>0\Big\vert{\cal F}_{n-1}\right)>0\right\}.
  \end{eqnarray}
 \end{lemma}
 \begin{proof}
 It is enough to  prove the non-trivial equality $\{Z_{n-1} >0\}  =  \Gamma(n)$.  Indeed, due to\\ $E\left( P(\widetilde Z_n>0|{\cal F}_{n-1}) I_{\{Z_{n-1} = 0 \}} \right) = P(\widetilde{Z}_n>0=Z_{n-1}) = 0$, we get   $\Gamma(n)    \subset    \{Z_{n-1} >0\}.$ On the other hand, due to  $E\left( Z_{n-1} I_{\Gamma(n)^c} \right) = E\left( \widetilde{Z}_n\  I_{\Gamma(n)^c}\right) \leq  E\left( I_{\{\widetilde{Z}_n>0\}} \  I_{\Gamma(n)^c}\right)= 0$, we obtain  $  \{Z_{n-1} >0\} \subset  \Gamma(n)$. This ends the proof of the lemma.
 \end{proof}
\begin{lemma}\label{DISZZRhavesomezero}
 Let $R$ be an equivalent probability to $P$. Then the following hold for all $n$.
 \begin{eqnarray*}
  \{\widetilde{Z}_n = 0 \} = \{\widetilde{Z}_{n}^R= 0 \}, \ \ \ \mbox{and} \ \ \ \{Z_{n-1} = 0\}=\{Z^R_{n-1} = 0\},
 \end{eqnarray*}
 where $\widetilde{Z}^R_n := R(\tau \geq n| {\cal F}_n)$ and $ {Z}^R_{n-1}:=R(\tau \geq n| {\cal F}_{n-1})$.
 \end{lemma}
\begin{proof}
 Since
 \begin{eqnarray*}
    E\left[\widetilde{Z}_n I_{\{\widetilde{Z}^R_n=0\}}\right] = E\left[ I_{\{\tau \geq n\}} I_{\{\widetilde{Z}^R_n=0\}}\right] = 0, \ \mbox{and} \  E\left[ {Z}_{n-1} I_{\{ {Z}^R_{n-1}=0\}}\right] = E\left[ I_{\{\tau \geq n\}} I_{\{\widetilde{Z}^R_{n-1}=0\}}\right] = 0,
  \end{eqnarray*}
we obtain $\{\widetilde{Z}_n^R=0\} \subset \{\widetilde{Z}_n = 0\}$ and $\{Z^R_{n-1} = 0\}\subset \{Z_{n-1} = 0\}$.  The symmetric roles of $R$ and $P$ complete the proof of the lemma.
\end{proof}
\begin{proposition}\label{INDISimportprop1}
Let $X$ be an $\mathbb{F}$-martingale. Then the following are equivalent.\\
$\rm (a)$ For all $n$, we have
\begin{eqnarray}\label{INDISXZzero}
 E\left( \Delta X_n I_{\{\widetilde{Z}_n = 0\}}  \Big \vert {\cal F}_{n-1}\right) = 0.
\end{eqnarray}
$\rm (b)$  $X^\tau$ is a $\mathbb{G}$-martingale  under the probability $\mathbb{Q}:= \prod_{n=1}^N q_n$, where
$$
q_n:= \left(\frac{Z_{n-1}}{\widetilde{Z}_n}I_{\{n\leq \tau\}} + I_{\{n> \tau\}} \right)\left(  P\left( \widetilde{Z}_n >0 | {\cal F}_{n-1}\right)I_{\{n\leq \tau\}}  +  I_{\{n> \tau\}}\right)^{-1}.
$$
\end{proposition}
\begin{proof}
First, we remark that the probability $\mathbb{Q}$ is well defined since $ P\left( \widetilde{Z}_n >0 | {\cal F}_{n-1}\right)I_{\{n\leq \tau\}} +  I_{\{n> \tau\}}  >0$ due to Lemma \ref{DIStheincluseion111}. To complete the proof, we calculate that
\begin{eqnarray*}
 &&\left(  P\left( \widetilde{Z}_n >0 | {\cal F}_{n-1}\right)I_{\{n\leq \tau\}} +  I_{\{n> \tau\}}\right)E^{\mathbb{Q}}\left(  \Delta X_n I_{\{n\leq \tau\}} \Big\vert {\cal G}_{n-1}\right) \\
 && = E \left(  \Delta X_n I_{\{\widetilde{Z}_n >0\}} \Big\vert {\cal F}_{n-1}\right) I_{\{n\leq \tau\}}= - E \left(  \Delta X_n I_{\{\widetilde{Z}_n =0\}} \Big\vert {\cal F}_{n-1}\right) I_{\{n\leq \tau\}}.
\end{eqnarray*}
This ends the proof of the proposition.
\end{proof}
Throughout the rest,  we consider the following notation.
\begin{eqnarray}
\mathbb{Q}^{(e)}&:=& \prod_{n=1}^N \left(  \frac{\widetilde{Z}_n}{Z_{n-1}} I_{\{Z_{n-1} >0\}}  + I_{\{\widetilde{Z}_n=0<Z_{n-1}\}} +I_{\{Z_{n-1} = 0 \}}\right)\left( 1 + E\left(  I_{\{\widetilde{Z}_n=0<Z_{n-1}\}} \Big\vert {\cal F}_{n-1}\right) \right)^{-1}\centerdot \mathbb{P} \sim \mathbb{P}. \nonumber
\end{eqnarray}
Below, we state the main theorem in this subsection which shows what we can conclude if the market $(\mathbb{G},X^\tau)$ excludes arbitrage opportunities for any $\mathbb{F}$-adapted  integrable process $X$.
\begin{theorem}\label{theo:reversemainbefore}
 Let $\tau$ be a random time and $X$ be an $\mathbb{F}$-adapted  integrable process.  Then the following are equivalent. \\
 $\rm (a)$  $X^\tau$ satisfies NA$(\mathbb{G},\mathbb{P})$.\\
 $\rm (b)$   $X^{(e)}$ satisfies NA$(\mathbb{F},\mathbb{P})$, where  $\Delta X^{(e)}_n := \Delta X_n I_{\{\widetilde{Z}_{n}>0\}}$.
\end{theorem}

\begin{proof}
(a)$\Longrightarrow$(b).  If  $X^\tau$ satisfies NA$(\mathbb{G})$, there exists a probability $\mathbb{Q}^\mathbb{G}:=\prod_{n=1}^N (1+\Delta K^\mathbb{G}_n)\centerdot\mathbb{P} \sim \mathbb{P}$ such that $X^\tau$ is a $(\mathbb{G},\mathbb{Q}^\mathbb{G})$-martingale, where $1+\Delta K^\mathbb{G}_n >0$ and $E\left( 1+\Delta K^\mathbb{G}_n | {\cal G}_{n-1} \right) = 1$,  for all $n$.
 By Jeulin  \cite{Jeu}, there exists two  ${\cal{F}}_n$-measurable random variables $Y^\mathbb{F}_n$ and  $\phi_n$ such that
 \begin{eqnarray}\label{INDISKG}
  \left(1+\Delta K^\mathbb{G}_n\right)I_{\{n\leq \tau\}} = Y^\mathbb{F}_n I_{\{n \leq \tau\}}   + \phi_n  I_{\{n =  \tau\}}.
 \end{eqnarray}
 Therefore,
  \begin{eqnarray*}
  I_{\{n\leq \tau\}}  &=& E\left(\left(1+\Delta K^\mathbb{G}_n\right)I_{\{n\leq \tau\}} \Big\vert {\cal G}_{n-1}\right) = E\left( Y^\mathbb{F}_n I_{\{n\leq \tau\}}   + \phi_n  I_{\{n =  \tau\}} \Big\vert {\cal G}_{n-1}\right) \\
  &=& \frac{I_{\{n\leq \tau\}}}{Z_{n-1}}E\left( Y_n^\mathbb{F}\widetilde{Z}_n  + \phi_n \Delta  D^{o,\mathbb{F}}_n \ \Big\vert {\cal F}_{n-1}  \right) \\
  &=&  \frac{I_{\{n\leq \tau\}}}{Z_{n-1}}E\left( \widetilde{Z}_n \left( Y_n^\mathbb{F} + \frac{  \phi_n \Delta  D^{o,\mathbb{F}}_n}{\widetilde{Z}_n}  I_{\{\widetilde{Z}_n >0\}}\right) \ \Big\vert {\cal F}_{n-1}  \right),
 \end{eqnarray*}
 where  $\Delta  D^{o,\mathbb{F}}_n  = \widetilde{Z}_n - Z_n$ and $\Delta  D^{o,\mathbb{F}}_n I_{\{\widetilde{Z}_n = 0\}} = (\widetilde{Z}_n - Z_n)I_{\{\widetilde{Z}_n = 0\}} = 0$.  Hence, we get
  \begin{eqnarray*}
   &&E\left( \frac{\widetilde{Z}_n I_{\{Z_{n-1} >0\}}}{Z_{n-1}} \  \widetilde{Y}_n I_{\{\widetilde{Z}_n > 0\}} \ \Big\vert {\cal F}_{n-1}  \right) = I_{\{Z_{n-1} >0\}}, \ \mbox{and}\\
   &&\widetilde{Y}_n:= Y_n^\mathbb{F}+ \frac{  \phi_n \Delta  D^{o,\mathbb{F}}_n}{\widetilde{Z}_n} >0, \ \ \mbox{ on the set} \ \ \{\widetilde{Z}_n >0\}.
  \end{eqnarray*}
 Define $L$  by
\begin{eqnarray*}
 L_k&:=& \prod_{n=1}^k \left(    \widetilde{Y}_n I_{\{\widetilde{Z}_n >0\}} + I_{\{\widetilde{Z}_n = 0<Z_{n-1}\}}    + I_{\{Z_{n-1} = 0\}}\right) >0.
\end{eqnarray*}
It is easy to check that $L$ is an   $(\mathbb{F},\mathbb{Q}^{(e)})$-martingale, i.e. $E^{Q^{(e)}}\left( \frac{L_n}{L_{n-1}} \Big\vert {\cal F}_{n-1} \right) = 1$ for all $1\leq n\leq N$.   Since $X^\tau$ is a  $(\mathbb{G},\mathbb{Q}^\mathbb{G})$-martingale, due to (\ref{INDISKG}), we deduce that
\begin{eqnarray}
 0&=& E\left( \Delta X_n I_{\{n\leq \tau\}} \left( Y^\mathbb{F}_n I_{\{n\leq \tau\}}  + \phi_n  I_{\{n =  \tau\}} \right) \Big\vert {\cal G}_{n-1}\right) \nonumber \\
 &=&  E\left( \Delta X_n \left(\widetilde{Z}_n Y^\mathbb{F}_n  +  \phi_n  \Delta  D^{o,\mathbb{F}}_n\right) \  \Big\vert {\cal F}_{n-1}\right) \frac{I_{\{n\leq \tau\}}}{Z_{n-1}}. \nonumber
\end{eqnarray}
Hence, by taking conditional expectation under ${\cal F}_{n-1}$ in the above equality and using the fact $\{\widetilde{Z}_n >0\} \subset \{Z_{n-1} >0\}$, we get
\begin{eqnarray}\label{NADISGtoFzero}
 E\left( \Delta X_n  \frac{\widetilde{Z}_n I_{\{\widetilde{Z}_{n} > 0\}}}{Z_{n-1}} \widetilde{Y}_n  \  \Big\vert {\cal F}_{n-1}\right)  = 0.
\end{eqnarray}
Then, we deduce
\begin{eqnarray*}
\left( 1 + E\left(  I_{\{\widetilde{Z}_n=0<Z_{n-1}\}} \Big\vert {\cal F}_{n-1}\right) \right) E^{Q^{(e)}}\left( \Delta X^{(e)} \frac{L_n}{L_{n-1}} \Big\vert {\cal F}_{n-1}\right) = E\left( \Delta X_n  \frac{\widetilde{Z}_n I_{\{\widetilde{Z}_{n} > 0\}}}{Z_{n-1}} \widetilde{Y}_n  \  \Big\vert {\cal F}_{n-1}\right)=0.
\end{eqnarray*}
Therefore,   $LX^{(e)}$ is an $\mathbb{F}$-martingale under   $\mathbb{Q}^{(e)}$ and $X^{(e)}$ satisfies NA$(\mathbb{F},Q^{(e)})$ and NA$(\mathbb{F},\mathbb{P})$.\\

\noindent (b)$\Longrightarrow$(a).  Since $X^{(e)}$ satisfies NA$(\mathbb{F},\mathbb{P})$, there exists a probability $\mathbb{R}$ equivalent to $\mathbb{P}$ such that $X^{(e)}$ is an $(\mathbb{F},\mathbb{R})$-martingale. By Lemma \ref{DISZZRhavesomezero}, the condition (\ref{INDISXZzero}) in Proposition \ref{INDISimportprop1} is trivial satisfied by $X^{(e)}$ under the probability $\mathbb{R}$, i.e. $E^\mathbb{R}\left( \Delta X^{(e)}_n I_{\{\widetilde{Z}^\mathbb{R}_n = 0\}}  \Big \vert {\cal F}_{n-1}\right) = E^\mathbb{R}\left( \Delta X^{(e)}_n I_{\{\widetilde{Z}_n = 0\}}  \Big \vert {\cal F}_{n-1}\right) = 0.$ Therefore, by Proposition \ref{INDISimportprop1}, we conclude that $\left(X^{(e)}\right)^\tau = X^\tau$ satisfies NA$(\mathbb{G},\mathbb{P})$.\\
This ends the proof of the theorem.
\end{proof}

\begin{remark}
  The related work  in Aksamit et al. \cite{cdm2013} and Choulli et al. \cite{cadj20141}  proved similar   and  other much deeper results in continuous time semi-martingale settings by using optional stochastic integral and predictable characteristics. Let us point out that the $\mbg$-martingale $N^\mbg$ in Proposition \ref{prop:density} could be also written as a discrete time version of optional stochastic integral.
\end{remark}

\section{Non-arbitrage on $\Rbrack \tau, +\infty\Lbrack$}\label{sec:aftertau}
In this section, we shall move to the stability of non-arbitrage after an honest time $\tau$. We recall its definition below.
\begin{definition}
  A random time $\tau$ is honest, if for any $n$, there exists an ${\cal {F}}_n$-measurable r.v. $\tau _n$ such
that $\tau \mathds{1}_{\{\tau<n\}}=\tau_n \mathds{1}_{\{\tau<n\}}$.
\end{definition}
For more details on honest times, we consult Jeulin \cite{Jeu} and Barlow \cite{Barlow}.
For an honest time $\tau$, we associate the following stopping times
\beq\label{eq:crucialstoppingtimehonest}
\sigma_1:=\inf\{n\geq 1:  Z_n < 1\}, \ \ \sigma_2:=\inf\{n\geq 1:  Z_{n-1} < 1\} \ \mbox{ and } \sigma_3:=\inf\{n\geq 1:  \widetilde{Z}_n < 1\}. \ \   \ \
\eeq

\begin{lemma}
  For an honest time $\tau$ and  stopping times in (\ref{eq:crucialstoppingtimehonest}), the following hold for all  $1\leq n\leq N$.\\
  \noindent $\rm (a)$  $\{\widetilde{Z}_n <1\}\subset \{Z_{n-1} <1\}$ and $\{\widetilde{Z}_n <1\}\subset \{Z_n <1\}$.\\
  \noindent $\rm (b)$ $\sigma_2$ is an $\mbf$-predictable stopping time and $    \sigma_2 \leq \sigma_3$ and $ \sigma_1 \leq \sigma_3.$\\
 \noindent $\rm (c)$ $\tau \geq  \sigma_1$ and  $Z_{n-1}, \widetilde{Z}_n <1$ on $\{\tau < n\}$.
\end{lemma}
\begin{proof}
(a) Notice that
 \beq
  E\left[\left(1 - \widetilde{Z}_n\right) \mathds{1}_{\{Z_{n-1} = 1\}}\right] = E\left[\left(1 - Z_{n-1} \right) \mathds{1}_{\{Z_{n-1} = 1\}}\right] = 0.
  \eeq
  Hence, $\{Z_{n-1} = 1\} \subset \{\widetilde{Z}_n = 1\}$.  Due to $Z_n\leq \widetilde{Z}_n$, we have  $\{\widetilde{Z}_n <1\}\subset \{Z_n <1\}$. \\

  (b) Since $\{\sigma_2\leq n\} = \{Z_{n-1} < 1\}\in {\cal F}_{n-1}$, we conclude that  $\sigma_2$ is predictable.  The inequalities $    \sigma_2 \leq \sigma_3$ and $ \sigma_1 \leq \sigma_3$ follow immediately from (a). \\

  (c) Notice that
  \begin{eqnarray*}
   E \left[\ \mathds{I}_{\{n > \tau\}} \mathds{I}_{\{Z_{n-1} = 1\}} \right] &=& E \left[ (1-Z_{n-1}) \mathds{I}_{\{Z_{n-1} = 1\}} \right]  = 0,  \ \mbox{and}\\
   E \left[\ \mathds{I}_{\{n > \tau\}} \mathds{I}_{\{\widetilde{Z}_{n} = 1\}} \right] &=& E \left[ (1-\widetilde{Z}_{n})\mathds{I}_{\{\widetilde{Z}_{n} = 1\}} \right]  = 0.
  \end{eqnarray*}
  Therefore, $Z_{n-1}<1$ and $\wid{Z}_n<1$   on the set $\{n > \tau\}$. This ends the proof of the lemma.
 \end{proof}
The following lemma describes the connection between conditional expectations under $\mbf$ and $\mbg$. For its proof, we consult Jeulin \cite{Jeu}.
\begin{lemma}\label{lem:OPunderGFafter}
Let $Y$ be an integrable $\mb A$-measurable random variable. Then, the following  hold.\\
$\rm (a)$ On the set $\{n>\tau\}$, the conditional expectation under ${\cal G}_n$ is given by
\beq
    E\left[Y|{\cal G}_n\right]\mathds{1}_{\{\tau <n\}} = E\left[Y \mathds{1}_{\{\tau<n\}} |{\cal F}_n\right] \frac{1}{1 - \widetilde{Z}_n}\mathds{1}_{\{\tau < n\}}.
    \eeq
$ \rm (b)$ On the set $\{n>\tau\}$, the conditional expectation under ${\cal G}_{n-1}$ is given by
\beq
    E\left[Y|{\cal G}_{n-1}\right]\mathds{1}_{\{\tau <n\}} = E\left[Y \mathds{1}_{\{\tau < n\}} |{\cal F}_{n-1}\right] \frac{1}{1 - Z_{n-1}}\mathds{1}_{\{\tau < n\}}.
    \eeq
Moreover, if $Y$ is ${\cal F}_n$-measurable, we have
\beq\label{eq:OPunderGFafter}
E\left[Y|{\cal G}_{n-1}\right]\mathds{1}_{\{\tau < n\}} &=& E\left[Y (1 - \widetilde{Z}_n) |{\cal F}_{n-1}\right] \frac{1}{ 1 - Z_{n-1}}\mathds{1}_{\{\tau < n\}}.
\eeq
\end{lemma}
The following theorem characterizes the relationship between $\mathbb F$-martingales and $\mathbb G$-martingales on the stochastic interval $\Rbrack \tau, +\infty \Lbrack$. For the continuous time settings, we consult Jeulin \cite{Jeu}.
\begin{theorem}\label{theo:GFmaraftertau}
  Let $M$ be an $\mbf$-martingale and $\tau$ be an honest  time. Then the following process
  \beq
  \widehat{M}^{(a)}_n := M_{n\vee \tau} - M_{\tau} - \sum_{1\leq k\leq n} \frac{1}{1 - Z_{k-1}}\mathds{1}_{\{\tau < k\}} E\left[\Delta M_k (1 - \widetilde{Z}_k) | \cf_{k-1}\right],\nonumber
  \eeq
  is a $\mbg$-martingale.
\end{theorem}
\begin{proof}
  Although it can be derived from Jeulin \cite{Jeu}, we opt to give a direct proof here. To this end, by using Lemma \ref{lem:OPunderGFafter}-(b),  we calculate
  \beq
  E\Big[\Delta \widehat{M}^{(a)}_{n} \Big| \cg_{n-1}\Big] &=& E\Big[\Delta M_{n}\mathds{1}_{\{\tau <n\}} -   \frac{1}{1 - Z_{n-1}}\mathds{1}_{\{\tau < n\}} E\left[\Delta M_n  (1 - \widetilde{Z}_n) | \cf_{n-1}\right] \Big| \cg_{n-1}\Big]\nonumber \\
  &=&  E\Big[\Delta M_{n}\mathds{1}_{\{\tau <n\}}  \Big| \cg_{n-1}\Big] -   \frac{1}{1 - Z_{n-1}}\mathds{1}_{\{\tau < n\}}E\left[   \Delta M_n  (1 - \widetilde{Z}_n) | \cf_{n-1} \right] \nonumber \\
  &=& 0. \nonumber
  \eeq
  This ends the proof of theorem.
 \end{proof}

The following proposition is constructing a $\mathbb G$-martingale density for  a class of $\mathbb{G}$-semi-martingales.
\begin{proposition}\label{prop:densityafter}
The following process
 \beq\label{eq:mardensityafter}
  \widehat{N}^{(a)}_n:= -\sum_{1\leq k\leq n} \mathds{1}_{\{\tau <k\}}E[\mathds{1}_{\{\widetilde{Z}_{k} <1\}} | {\cal F}_{k-1}] + \sum_{1\leq k\leq n} \frac{1 - Z_{k-1}}{1 - \widetilde{Z}_k}\mathds{1}_{\{\tau < k\}}
  \eeq
  is a $\mbg$-martingale such that $1 + \Delta \widehat{N}^{(a)}_n >0$.
\end{proposition}

\begin{proof}
  First, we prove that $\widehat{N}^{(a)}$ is a $\mb G$-martingale. To this end, by using Lemma \ref{lem:OPunderGFafter}-(b), we calculate
  \bq
  E\left[\Delta {\widehat{N}}^{(a)}_{n} | {\cal G}_{n-1}\right] &=& - \mathds{1}_{\{\tau < n\}}E[\mathds{1}_{\{\widetilde{Z}_{n} <1\}} | {\cal F}_{n-1}] +  E\left[  \frac{1 - Z_{n-1}}{1 - \widetilde{Z}_n}\mathds{1}_{\{\tau < n\}} \Big| {\cal G}_{n-1}\right]   \\
  &=& - \mathds{1}_{\{\tau < n\}}E[\mathds{1}_{\{\widetilde{Z}_{n} <1\}} | {\cal F}_{n-1}] +   \mathds{1}_{\{\tau < n\}}E[\mathds{1}_{\{\widetilde{Z}_{n} <1\}} | {\cal F}_{n-1}] = 0.
  \eq 
  \noindent Next, we show that $1 + \Delta {\widehat{N}}^{(a)}_n >0$. Indeed
  \bq
  1 + \Delta {\widehat{N}}^{(a)}_n = 1  -\mathds{1}_{\{\tau < n\}}E[\mathds{1}_{\{\widetilde{Z}_{n} <1\}} | {\cal F}_{n-1}] + \frac{1 - Z_{n-1}}{1 - \widetilde{Z}_n} \mathds{1}_{\{\tau < n\}} \geq  \mathds{1}_{\{\tau \geq  n\}} + \frac{1 - Z_{n-1}}{1 - \widetilde{Z}_n} \mathds{1}_{\{\tau < n\}} >0.
  \eq
  The integrability of ${\widehat{N}}^{(a)}$ follows from the fact that $E\left\vert {\widehat{N}}^{(a)}_n\right\vert \leq 2n$.   This completes the proof of the proposition.
 \end{proof}
Below, we state the first main theorem of this section.
\begin{theorem}\label{theo:maintheoaftertau}
  Consider an honest time $\tau$  and an $\mathbb F$-martingale $S$. Denote the probability measure $\mathbb{Q}^{(a)} \sim \mathbb{P}$ with density $D^{(a)}_n:={\cal E}(Y^{(a)})_n$ where
  \beq\label{eq:crucialmeasureQaftertau}
  \Delta Y^{(a)}_n := (1 - \widetilde{Z}_{n} )\mathds{1}_{\{Z_{n-1} <1\}}E\left[\mathds{1}_{\{\widetilde{Z}_{n} =1\}} | {\cal F}_{n-1}\right] - (1-Z_{n-1})\mathds{1}_{\{\widetilde{Z}_{n} =1\}}, \ Y^{(a)}_0 = 0.
  \eeq
  Then the following are equivalent:\\
  $\rm (a)$ $S$ is a $(\mathbb{Q}^{(a)}, \mathbb F)$-martingale;\\
  $\rm (b)$ $S$ is orthogonal to $D^{(a)}$ and $Y^{(a)}$;\\
  $\rm (c)$ ${\cal E}(\widehat{N}^{(a)})_n (S_n - S_{n\wedge \tau})$ is a $\mathbb G$-martingale, where $\widehat{N}^{(a)}$ is given by (\ref{eq:mardensityafter}).\\
  As a consequence, all the above three equivalent conditions imply \\
  $\rm (d)$ $S- S^\tau$ satisfies NA$(\mathbb G, \mathbb{P})$ and NA$(\mathbb G, \mathbb{Q}^{(a)})$.
\end{theorem}
\begin{proof}
First, we remark that  $Y^{(a)}$ is an $\mathbb F$-martingale and $1 + \Delta Y^{(a)} > 0$. Indeed,  $$
1 + \Delta Y^{(a)}_n = (1 - \widetilde{Z}_n)\mathds{1}_{\{Z_{n-1}<1\}} E\left[\mathds{1}_{\{\widetilde{Z}_n=1\}} | {\cal F}_{n-1}\right] + \mathds{1}_{\{\widetilde{Z}_n<1\}} + Z_{n-1} \mathds{1}_{\{\widetilde{Z}_n=1\}} > 0,
$$
where we used the fact that on the set $\{\widetilde{Z}_n<1\}$, $1 + \Delta Y^{(a)}_n \geq 1$ and the inclusion  $\{\widetilde{Z}_n=1\} \subset \{Z_{n-1} >0\}$, since $\{Z_{n-1} =0\}\subset \{\widetilde{Z}_n=0\} $.  Therefore, $D^{(a)}$ is a strictly positive martingale.

\noindent The equivalence between (a) and (b) is obvious. In the following, we are trying to prove the equivalence between (a) and (c).  Recall that
  \beq
  {\widehat{N}}^{(a)}_n = -\sum_{1\leq k\leq n} \mathds{1}_{\{\tau < k\}}E[\mathds{1}_{\{\widetilde{Z}_{k} <1\}} | {\cal F}_{k-1}] + \sum_{1\leq k\leq n} \frac{1- Z_{k-1}}{1- \widetilde{Z}_k}\mathds{1}_{\{\tau < k\}}.
  \eeq
   Due to Lemma \ref{lem:OPunderGF} , we deduce that
  \bq
  E\Big[\frac{\Delta S_k }{1 - \widetilde{Z}_k}  \mathds{1}_{\{\tau < k\}} | {\cal G}_{k-1}\Big] &=&  \frac{\mathds{1}_{\{\tau < k\}}}{1- Z_{k-1}}E\Big[\Delta S_k \mathds{1}_{\{\widetilde{Z}_k <1\}}  | {\cal F}_{k-1}\Big],\nonumber \\
  E\Big[ \Delta S_k  \mathds{1}_{\{\tau < k\}} | {\cal G}_{k-1}\Big] &=& \frac{\mathds{1}_{\{\tau < k\}}}{1- Z_{k-1}}E\Big[ \Delta S_k  (1-\widetilde{Z}_k) | {\cal F}_{k-1}\Big].
  \eq
  Since  $  S_{n+1} - S_{(n+1)\wedge \tau} = S_{n} - S_{n\wedge \tau}  + \Delta S_{n+1}\mathds{1}_{\{n+1>\tau\}}$, we derive  that
  \bq
  &&E\left[{\cal E}({\widehat{N}}^{(a)})_{n+1} S_{(n+1)\wedge \tau} | {\cal G}_{n+1} \right] \\
  && ={\cal E}({\widehat{N}}^{(a)})_n E\Big[ (1+\Delta {\widehat{N}}^{(a)}_{n+1}) S_{(n+1)\wedge \tau} | {\cal G}_{n} \Big] \\
  && = {\cal E}({\widehat{N}}^{(a)})_n E\Big[ S_n - S_{n\wedge \tau} + \Delta S_{n+1}\mathds{1}_{\{n+1> \tau\}}  + \Delta S_{n+1} \Delta {\widehat{N}}^\mbg_{n+1} \mathds{1}_{\{n+1 > \tau\}}| {\cal G}_{n} \Big] \\
  &&={\cal E}({\widehat{N}}^{(a)})_n \left\{ S_n - S_{n\wedge \tau} + E\Big[ \Delta S_{n+1}(1 -\widetilde{Z}_{n+1})E\left[\mathds{1}_{\{\widetilde{Z}_{n+1}=1 \}}| {\cal F}_{n}\right] | {\cal F}_n \Big]\frac{\mathds{1}_{\{n+1> \tau\}}}{1 - Z_n}  \right\}\\
  && - {\cal E}({\widehat{N}}^{(a)})_n \left\{ \mathds{1}_{\{n+1> \tau\}} E\Big[ \Delta S_{n+1} \mathds{1}_{\{\widetilde{Z}_{n+1}=1\}} | {\cal F}_n \Big]\right\}\\
  &&= {\cal E}({\widehat{N}}^{(a)})_n  \left( S_n - S_{n\wedge \tau} \right) \\
  && + {\cal E}({\widehat{N}}^{(a)})_n \left\{E\left[\Delta S_{n+1}\left\{(1 - \widetilde{Z}_{n+1} )E\left[\mathds{1}_{\{\widetilde{Z}_{n+1} =1\}} | {\cal F}_n\right] - (1-Z_n)\mathds{1}_{\{\widetilde{Z}_{n+1} =1\}} \right\} | {\cal F}_n\right]\right\}\frac{\mathds{1}_{\{n+1> \tau\}}}{1 - Z_n}\\
  &&= {\cal E}({\widehat{N}}^{(a)})_n  (S_n - S_{n\wedge \tau}) +   {\cal E}({\widehat{N}}^{(a)})_n E^{\mathbb{Q}^{(a)}}\left[\Delta S_{n+1} | {\cal F}_n\right]\frac{\mathds{1}_{\{n+1> \tau\}}}{1 - Z_n}.
  \eq
  Therefore, (a) implies (c).  Conversely, if (c) holds, we have $$E^{\mathbb{Q}^{(a)}} \left[\Delta S_{n+1} | {\cal F}_n\right]\frac{\mathds{1}_{\{n+1 > \tau\}}}{1 - Z_n} = 0, \ \mbox{and }  \ E^{\mathbb{Q}^{(a)}} \left[\Delta S_{n+1} | {\cal F}_n\right]\mathds{1}_{\{Z_n<1\}} = 0.$$ Notice that $$E^{\mathbb{Q}^{(a)}}\left[\Delta S_{n+1} | {\cal F}_n\right]\mathds{1}_{\{Z_n=1\}} = 0, \mbox{ for all $n$ }.$$ Thus, we conclude that $E^{\mathbb{Q}^{(a)}} \left[\Delta S_{n+1} | {\cal F}_n\right] = 0$, for all $n$.
  This ends the proof of the theorem.
   \end{proof}

\begin{remark}
  We observe from Theorem  \ref{theo:maintheoaftertau} that even though $Y^{(a)}$ is an $\mathbb F$-martingale, the  process $Y^{(a)}_n - Y^{(a)}_{n\wedge \tau} =  \sum_{k\leq n}(1 - \widetilde{Z}_{k} )E\left[\mathds{1}_{\{\widetilde{Z}_{k} =1\}} | {\cal F}_k\right] \mathds{1}_{\{k>\tau\}}$ fails  NA$(\mathbb G)$ since it is a $\mathbb G$-increasing process.
\end{remark}
%

\begin{corollary}\label{theo:naafterhonest}
  For any $\mbf$-martingale $M$, if for all $n$
  \beq\label{eq:crucialassumptionhonest}
  \{\widetilde{Z}_n = 1\} = \{Z_{n-1} = 1\}.
  \eeq
  Then the following properties hold:\\
  $\rm (a) $ The process $M_{n} - M_{n\wedge \tau}$ satisfies NA$(\mbg)$;\\
  $\rm (b) $  $\left({\cal E}(\widehat{N}^{(a)})_n \left(M_{n} - M_{n\wedge \tau}\right)\right)_{n\geq 1}$ is a $\mbg$-martingale, where $\widehat{N}^{(a)}$ is given by (\ref{eq:mardensityafter}) in Proposition \ref{prop:densityafter};\\
  $\rm (c) $ The probability measure $\mathbb{Q}^{(a)}$,  given in (\ref{eq:crucialmeasureQaftertau}), coincides with $\mathbb{P}$.
\end{corollary}

Below, we state our second main theorem in this section,   where we give the necessary and sufficient conditions that imposed on the random time $\tau$ (or the stopping times in (\ref{eq:crucialstoppingtimehonest})) to guarantee the   process $M - M^\tau$ satisfies NA($\mbg$) for any $\mbf$-martingale $M$.

\begin{theorem}
Consider an honest  time $\tau$ and the associated stopping times defined in (\ref{eq:crucialstoppingtimehonest}). Then the following are equivalent:\\
$\rm (a)$ For any $\mbf$-martingale $M$, the  process $M_{n} - M_{n\wedge \tau}$ satisfies NA$(\mbg)$.\\
$\rm (b)$ $ \{\widetilde{Z}_n = 1\} =\{Z_{n-1} = 1\}$ for all $n$. \\
$\rm (c)$ $\sigma_1 + 1=  \sigma_2 =\sigma_3$.\\
$\rm (d)$ $\sigma_3$ is an $\mathbb{F}$-predictable stopping time.\\
$\rm (e)$ The probability $\mathbb{Q}^{(a)}$, defined via (\ref{eq:crucialmeasureQaftertau}), coincides with $\mathbb{P}$.
\end{theorem}
\begin{proof}
The proof of the theorem would be achieved after four steps. In the first step, we prove (b)$\Leftrightarrow$(c). The second step  focuses on (b)$\Leftrightarrow$(d). The third step deals with (b)$\Leftrightarrow$(e). In the last step, we prove (a) $\Leftrightarrow$ (b).  \\

\noindent Step 1:  The equivalence between (b) and (c) is obvious. Indeed, if (b) holds, it is trivial that $\sigma_2 = \sigma_3$. Conversely, if (c) holds, we derive that $$E\left( (1 - Z_{n-1}) I_{\{\widetilde{Z}_n = 1\}} \right) = E\left( (1 - Z_{n-1}) I_{\{\widetilde{Z}_n = 1\}} I_{\{n < \sigma_3\}} \right)=E\left( (1 -Z_{n-1}) I_{\{\widetilde{Z}_n = 1\}} I_{\{n < \sigma_2\}} \right) = 0.$$
  Hence, we conclude that $\{\widetilde{Z}_n = 1\} \subset \{Z_{n-1} = 1\}$ for all $n$.  \\

  \noindent Step 2: We prove (b)$\Leftrightarrow$(d). If (b) holds, it is easy to see that  $\sigma_3$ is an $\mathbb{F}$-predictable stopping time. Conversely, due to the predictability of $\sigma_3$ and $$ E\left[(1 - Z_{n-1}) I_{\{\widetilde{Z}_n=1\} }\right] = E\left[(1 - Z_{n-1})  I_{\{n<\sigma_3\}  }\right]=E\left[(1 - \widetilde{Z}_n)  I_{\{n<\sigma_3\}  }\right]=0,$$
  we conclude that $\{\widetilde{Z}_n = 1\} \subset \{Z_{n-1} =1\}$ for all $n$.\\

   \noindent Step 3:  We prove (b)$\Leftrightarrow$(e). If (b) holds, apparently, $Y^{(a)}=0$ and $\mathbb{Q}^{(a)}=\mathbb{P}$. Conversely, if (e) holds, $\Delta Y_n^{(a)} = 0$ for all $n$. Hence, $(1-\widetilde{Z}_{n})\mathds{1}_{\{Z_{n-1}<1\}} E\left[\mathds{1}_{\{\widetilde{Z}_{n} =1\}} | {\cal F}_{n-1}\right] = (1-Z_{n-1}) \mathds{1}_{\{\widetilde{Z}_{n}=1 \}}=0$ and $\{\widetilde{Z}_{n}=1 \} = \{{Z}_{n-1}=1 \}$ for all $n$. \\

  \noindent Step 4: We prove (a)$\Leftrightarrow$(b).    Suppose for any $\mbf$-martingale $M$, the stopped process $M^\tau$ satisfies NA($\mbg$). Consider the $\mathbb F$-martingale
  \bq
  M_n := \sum_{1\leq k\leq n} \left(\mathds{1}_{\{\widetilde{Z}_k =1 \}} - E\left[\mathds{1}_{\{\widetilde{Z}_k =1 \}} | {\cal F}_{k-1}\right]\right).
  \eq
  It is easy to see that $M_n - M_{n\wedge \tau} = -\sum_{1\leq k\leq n}  E\left[\mathds{1}_{\{\widetilde{Z}_k =1 \}} | {\cal F}_{k-1}\right]\mathds{1}_{\{\tau <k\}}$. Note that $M_n - M_{n\wedge \tau}$ is a $\mathbb G$ predictable decreasing process satisfying NA($\mathbb G$). Therefore it is null. Then, we deduce that
  \bq
  0&=&E\left[M_n - M_{n\wedge \tau}\right] = \sum_{1\leq k\leq n} E\left[  E\left[\mathds{1}_{\{\widetilde{Z}_k =1 \}} | {\cal F}_{k-1}\right]\mathds{1}_{\{\tau <k\}}\right] \\
  &=&  \sum_{1\leq k\leq n} E\left[ (1-Z_{k-1}) \mathds{1}_{\{\widetilde{Z}_k =1 \}} \right].
  \eq
  Hence,  $\{\widetilde{Z}_k =1\} \subset \{Z_{k-1} =1\}$ for all $k$.
  \\
  The reverse implication follows immediately from Theorem \ref{theo:maintheoaftertau} or Corollary \ref{theo:naafterhonest}. This ends  the proof of the theorem.
 \end{proof}
\subsection{Reverse Problem: after $\tau$}
The previous section studied  what we can conclude for arbitrage opportunities  from the standpoint view of the insider. In this section, we will investigate the equivalence or consequence on the market $S$ if we know that the insider can not make arbitrage  in the market $(\mathbb{G},S-S^\tau)$.\\

We start with two simple lemmas and one proposition before proving Theorem \ref{theo:reversemainafter} below.
 \begin{lemma}\label{DIStheincluseion111after}
  The following hold.
  \begin{eqnarray}
   \{n> \tau\} \  \subset \ \{\widetilde{Z}_n<1\} \  \subset \   \{Z_{n-1} <1\}  = \left\{P\left(\widetilde Z_n<1\Big\vert{\cal F}_{n-1}\right)>0\right\}.
  \end{eqnarray}
 \end{lemma}
 \begin{proof}
 It is enough to  prove the non-trivial equality $\{Z_{n-1} <1\}  = \left\{P\left(\widetilde Z_n<1\Big\vert{\cal F}_{n-1}\right)>0\right\}$.  Indeed, due to $E\left( P(\widetilde Z_n<1|{\cal F}_{n-1}) I_{\{Z_{n-1} = 1 \}} \right) = P(\widetilde{Z}_n<1=Z_{n-1}) = 0$, we get   $\left\{P\left(\widetilde Z_n<1\Big\vert{\cal F}_{n-1}\right)>0\right\}    \subset    \{Z_{n-1} <1\}.$ On the other hand, due to
\begin{eqnarray*}
 E\left( (1-Z_{n-1}) I_{\left\{P\left(\widetilde Z_n<1\big\vert{\cal F}_{n-1}\right)=0\right\}} \right)  &=&E\left( (1-\widetilde{Z}_{n}) I_{\left\{P\left(\widetilde Z_n<1\big\vert{\cal F}_{n-1}\right)=0\right\}} \right) \\
 &\leq&  E\left( I_{\{\widetilde{Z}_n<1\}} \  I_{\left\{P\left(\widetilde Z_n<1\big\vert{\cal F}_{n-1}\right)=0\right\}}\right)= 0,
 \end{eqnarray*}
 we obtain  $  \{Z_{n-1} <1\} \subset  \left\{P\left(\widetilde Z_n<1\Big\vert{\cal F}_{n-1}\right)>0\right\}$. This ends the proof of the lemma.
 \end{proof}
\begin{lemma}\label{DISZZRhavesomezeroafter}
 Let $R$ be an equivalent probability to $P$. Then the following hold for all $n$.
 \begin{eqnarray*}
  \{\widetilde{Z}_n = 1 \} = \{\widetilde{Z}_{n}^R= 1 \}, \ \ \ \mbox{and} \ \ \ \{Z_{n-1} = 1\}=\{Z^R_{n-1} = 1\},
 \end{eqnarray*}
 where $\widetilde{Z}^R_n := R(\tau \geq n| {\cal F}_n)$ and $ {Z}^R_{n-1}:=R(\tau \geq n| {\cal F}_{n-1})$.
 \end{lemma}
\begin{proof}
 Since
 \begin{eqnarray*}
    E\left[(1- \widetilde{Z}_n) I_{\{\widetilde{Z}^R_n=1\}}\right] = E\left[ I_{\{\tau <n\}} I_{\{\widetilde{Z}^R_n=1\}}\right] = 0, \ \mbox{and} \  E\left[ (1-{Z}_{n-1}) I_{\{ {Z}^R_{n-1}=1\}}\right] = E\left[ I_{\{\tau < n\}} I_{\{\widetilde{Z}^R_{n-1}=1\}}\right] = 0,
  \end{eqnarray*}
we obtain $\{\widetilde{Z}_n^R=1\} \subset \{\widetilde{Z}_n = 1\}$ and $\{Z^R_{n-1} = 1\}\subset \{Z_{n-1} = 1\}$.  The symmetric roles of $R$ and $P$ complete the proof of the lemma.
\end{proof}
\begin{proposition}\label{INDISimportprop1after}
Let $X$ be an $\mathbb{F}$-martingale. Then the following are equivalent.\\
$\rm (a)$ For all $n$, we have
\begin{eqnarray}\label{INDISXZzeroafter}
 E\left( \Delta X_n I_{\{\widetilde{Z}_n =1\}}  \Big \vert {\cal F}_{n-1}\right) = 0.
\end{eqnarray}
$\rm (b)$  $X - X^\tau$ is a $\mathbb{G}$-martingale  under the probability $\mathbb{Q}:= \prod_{n=1}^N q^{(a)}_n$, where
$$
q^{(a)}_n:= \left(\frac{1 - Z_{n-1}}{1 - \widetilde{Z}_n}I_{\{n > \tau\}} + I_{\{n\leq \tau\}} \right)\left(  P\left( \widetilde{Z}_n <1 | {\cal F}_{n-1}\right)I_{\{n > \tau\}}  +  I_{\{n\leq  \tau\}}\right)^{-1}.
$$
\end{proposition}
\begin{proof}
First, due to Lemma \ref{DIStheincluseion111after},  we remark that the probability $\mathbb{Q}$ is well defined. To complete the proof, we calculate 
\begin{eqnarray*}
 &&\left(  P\left( \widetilde{Z}_n <1 | {\cal F}_{n-1}\right)I_{\{n > \tau\}} +  I_{\{n\leq \tau\}}\right)E^{\mathbb{Q}}\left(  \Delta X_n I_{\{n > \tau\}} \Big\vert {\cal G}_{n-1}\right) \\
 && = E \left(  \Delta X_n I_{\{\widetilde{Z}_n <1\}} \Big\vert {\cal F}_{n-1}\right) I_{\{n >\tau\}}= - E \left(  \Delta X_n I_{\{\widetilde{Z}_n =1\}} \Big\vert {\cal F}_{n-1}\right) I_{\{n > \tau\}}.
\end{eqnarray*}
This ends the proof of the proposition.
\end{proof}
Throughout the rest,  we consider the following notations.
\begin{eqnarray}
\widetilde{\mathbb{Q}}^{(e)}&:=& \prod_{n=1}^N \left(  \frac{1 - \widetilde{Z}_n}{1 - Z_{n-1}} I_{\{Z_{n-1} <1\}}  + I_{\{\widetilde{Z}_n=1>Z_{n-1}\}} +I_{\{Z_{n-1} = 1 \}}\right)\left( 1 + E\left(  I_{\{\widetilde{Z}_n=1>Z_{n-1}\}} \Big\vert {\cal F}_{n-1}\right) \right)^{-1}\centerdot \mathbb{P} \sim \mathbb{P}. \nonumber
\end{eqnarray}
Below, we state the main theorem in this subsection which shows what we can conclude if the market $(\mathbb{G},X- X^\tau)$ excludes arbitrage opportunities for any $\mathbb{F}$-adapted  integrable process $X$.
\begin{theorem}\label{theo:reversemainafter}
 Let $\tau$ be an $\mathbb{F}$-honest time and $X$ be an $\mathbb{F}$-adapted  integrable process.  Then the following are equivalent. \\
 $\rm (a)$  $X - X^\tau$ satisfies NA$(\mathbb{G},\mathbb{P})$.\\
 $\rm (b)$   $\widetilde{X}^{(e)}$ satisfies NA$(\mathbb{F},\mathbb{P})$, where  $\Delta \widetilde{X}^{(e)}_n := \Delta X_n I_{\{\widetilde{Z}_{n}<1\}}$.
\end{theorem}

\begin{proof}
(a)$\Longrightarrow$(b).  If  $X - X^\tau$ satisfies NA$(\mathbb{G})$, there exists a probability $\mathbb{Q}^\mathbb{G}:=\prod_{n=1}^N (1+\Delta K^\mathbb{G}_n)\centerdot\mathbb{P} \sim \mathbb{P}$ such that $X - X^\tau$ is a $(\mathbb{G},\mathbb{Q}^\mathbb{G})$-martingale, where $1+\Delta K^\mathbb{G}_n >0$ and $E\left( 1+\Delta K^\mathbb{G}_n | {\cal G}_{n-1} \right) = 1$,  for all $n$.
 By Proposition (5.3) in Jeulin  \cite{Jeu}, there exists an ${\cal{F}}_n$-measurable $Y^\mathbb{F}_n$ such that
 \begin{eqnarray}\label{INDISKGafter}
  \left(1+\Delta K^\mathbb{G}_n\right)I_{\{n > \tau\}} = Y^\mathbb{F}_n I_{\{n > \tau\}}.
 \end{eqnarray}
 Therefore,
  \begin{eqnarray*}
  I_{\{n > \tau\}}  &=& E\left(\left(1+\Delta K^\mathbb{G}_n\right)I_{\{n> \tau\}} \Big\vert {\cal G}_{n-1}\right) = E\left( Y^\mathbb{F}_n  I_{\{n> \tau\}} \Big\vert {\cal G}_{n-1}\right) \\
  &=& \frac{I_{\{n> \tau\}}}{1 - Z_{n-1}}E\left( Y_n^\mathbb{F} (1 - \widetilde{Z}_n)   \ \Big\vert {\cal F}_{n-1}  \right).
 \end{eqnarray*}
   Hence, we get
  \begin{eqnarray*}
   &&E\left( \frac{(1-\widetilde{Z}_n) I_{\{Z_{n-1} <1\}}}{1 - Z_{n-1}} \   {Y}^\mathbb{F}_n I_{\{\widetilde{Z}_n <1\}} \ \Big\vert {\cal F}_{n-1}  \right) = I_{\{Z_{n-1} <1\}} .
  \end{eqnarray*}
 Define $L$  by
\begin{eqnarray*}
 L_k&:=& \prod_{n=1}^k \left(    {Y}^\mathbb{F}_n I_{\{\widetilde{Z}_n <1\}} + I_{\{\widetilde{Z}_n = 1>Z_{n-1}\}}    + I_{\{Z_{n-1} = 1\}}\right) >0.
\end{eqnarray*}
It is easy to check that $L$ is an   $(\mathbb{F},\widetilde{\mathbb{Q}}^{(e)})$-martingale, i.e. $E^{\widetilde{Q}^{(e)}}\left( \frac{L_n}{L_{n-1}} \Big\vert {\cal F}_{n-1} \right) = 1$ for all $1\leq n\leq N$.   Since $X - X^\tau$ is a  $(\mathbb{G},\mathbb{Q}^\mathbb{G})$-martingale, due to (\ref{INDISKGafter}), we deduce that
\begin{eqnarray}
 0&=& E\left( \Delta X_n I_{\{n> \tau\}} Y^\mathbb{F}_n   \Big\vert {\cal G}_{n-1}\right)  =  E\left( \Delta X_n (1-\widetilde{Z}_n) Y^\mathbb{F}_n    \  \Big\vert {\cal F}_{n-1}\right) \frac{I_{\{n> \tau\}}}{1 - Z_{n-1}}. \nonumber
\end{eqnarray}
Hence, by taking conditional expectation under ${\cal F}_{n-1}$ in the above equality and using the fact $\{\widetilde{Z}_n <1\} \subset \{Z_{n-1} <1\}$, we get
\begin{eqnarray}\label{NADISGtoFzeroafter}
 E\left( \Delta X_n  \frac{(1-\widetilde{Z}_n) I_{\{\widetilde{Z}_{n} <1\}}}{1 - Z_{n-1}} {Y}^\mathbb{F}_n  \  \Big\vert {\cal F}_{n-1}\right)  = 0.
\end{eqnarray}
Then, we deduce
\begin{eqnarray*}
\left( 1 + E\left(  I_{\{\widetilde{Z}_n=1>Z_{n-1}\}} \Big\vert {\cal F}_{n-1}\right) \right) E^{\widetilde{Q}^{(e)}}\left( \Delta \widetilde{X}^{(e)} \frac{L_n}{L_{n-1}} \Big\vert {\cal F}_{n-1}\right) = E\left( \Delta X_n  \frac{(1-\widetilde{Z}_n) I_{\{\widetilde{Z}_{n} <1\}}}{1 - Z_{n-1}}  {Y}^\mathbb{F}_n  \  \Big\vert {\cal F}_{n-1}\right)=0.
\end{eqnarray*}
Therefore,   $L \widetilde{X}^{(e)}$ is an $\mathbb{F}$-martingale under   $\widetilde{\mathbb{Q}}^{(e)}$ and $\widetilde{X}^{(e)}$ satisfies NA$(\mathbb{F},\widetilde{\mathbb{Q}}^{(e)})$ and NA$(\mathbb{F},\mathbb{P})$.\\

\noindent (b)$\Longrightarrow$(a).  Since $\widetilde{X}^{(e)}$ satisfies NA$(\mathbb{F},\mathbb{P})$, there exists a probability $\mathbb{R}$ equivalent to $\mathbb{P}$ such that $\widetilde{X}^{(e)}$ is an $(\mathbb{F},\mathbb{R})$-martingale. By Lemma \ref{DISZZRhavesomezeroafter}, the condition (\ref{INDISXZzeroafter}) in Proposition \ref{INDISimportprop1after} is trivial satisfied by $\widetilde{X}^{(e)}$ under the probability $\mathbb{R}$, i.e. $E^\mathbb{R}\left( \Delta \widetilde{X}^{(e)}_n I_{\{\widetilde{Z}^\mathbb{R}_n = 1\}}  \Big \vert {\cal F}_{n-1}\right) = E^\mathbb{R}\left( \Delta \widetilde{X}^{(e)}_n I_{\{\widetilde{Z}_n = 1\}}  \Big \vert {\cal F}_{n-1}\right) = 0.$ Therefore, by Proposition \ref{INDISimportprop1after}, we conclude that $\widetilde{X}^{(e)} - \left(\widetilde{X}^{(e)}\right)^\tau = X - X^\tau$ satisfies NA$(\mathbb{G},\mathbb{P})$.\\
This ends the proof of the theorem.
\end{proof}


\section{Explicit Examples}\label{sec:exam}
In this section, we   revisit those two  examples  presented in  Introduction and calculate explicitly the Az\'ema supermartingales and the  arbitrage opportunities. For more examples, we refer to Aksamit et al. \cite{aksamit/choulli/deng/jeanblanc} and Fontana et al. \cite{fjs} for   continuous time settings.\\

{\bf Example \ref{example:1} (continued).}  We are now in the same settings as   Example \ref{example:1}. 
\begin{lemma} The following hold.\\
 $\rm (a)$ The progressive enlargement filtration $\mathbb{G} = ({\cal G}_n)_{0\leq n \leq 2}$ is given by
\begin{eqnarray*}
{\cal G}_0 &=& \{\emptyset, \Omega\}, \ {\cal G}_1 = \{ \emptyset, \Omega, \{\omega_1, \omega_2\}, \{\omega_3\}, \{ \omega_4\}\}, \ \  \mbox{and } \ \ {\cal G}_2 =  \sigma(\{\emptyset, \Omega, \{\omega_1 \}, \{\omega_2 \}, \{\omega_3 \}, \{\omega_4 \} \}).
\end{eqnarray*}
$\rm (b)$ ${\cal G}_{\tau-} = {\cal G}_1$ and ${\cal G}_{\tau} = {\cal G}_{\tau+} = {\cal G}_2$.
\end{lemma}
\begin{proof}
 (a) is trivial by the definition of $\mathbb{G}$. Notice that ${\cal G}_{\tau-} = {\cal G}_0\vee \sigma\{A\cap \{n<\tau\}, \ \ A\in {\cal G}_n, \ n=1,2\}={\cal G}_1$ and ${\cal G}_{\tau} = \sigma\{A, \ \ A\cap \{\tau \leq n\}\in {\cal G}_n, \ \ n=0,1,2\} = {\cal G}_2$. Then (b) follows.
\end{proof}

\begin{lemma}\label{lem:exam1--}
 For the above settings, the following properties hold.\\
 $\rm (a)$ The processes $A, m, Z$ and $\widetilde{Z}$ are given by
 \begin{eqnarray*}
  A_0 &=& 0, \ A_1 = p \mathds{1}_{\{\omega_3, \omega_4\}},  \ A_2 = p \mathds{1}_{\{\omega_3,  \omega_4\}} + \mathds{1}_{\{\omega_1, \omega_2, \omega_4\}}.  \\
  m_0 &=& 1, \ m_1 = 1, \ m_2 = p \mathds{1}_{\{\omega_3,  \omega_4\}} + \mathds{1}_{\{\omega_1, \omega_2, \omega_4\}}.  \\
  Z_0 &=& 1, \  Z_1= 1 - p \mathds{1}_{\{\omega_3, \omega_4\}}, \ Z_2 = 0. \\
  \widetilde{Z}_0 &=& 1, \ \widetilde{Z}_1 = 1, \  \widetilde{Z}_2 = \mathds{1}_{\{\omega_1, \omega_2, \omega_4\}}.
 \end{eqnarray*}
$\rm (b)$ $  \widetilde{Z}_\tau = 1, \   \tau = \sup \{n\geq 0 : \ \widetilde{Z}_n  = 1\},$   and $\tau$  {is an honest time}.\\
 $\rm (c) $ The stopping times   in (\ref{eq:crucialstoppingtime}) are  given by $$R_1 = 2, \  R_2=+\infty, \ R_3(\omega_3) = 2, \ R_3(\omega_1, \omega_2, \omega_4) = +\infty.$$
$\rm (d)$  The stopping times   in (\ref{eq:crucialstoppingtimehonest}) are given by
$$\sigma_1(\omega_1, \omega_2) = 2,  \sigma_1(\omega_3, \omega_4) = 1, \sigma_2(\omega_1, \omega_2) = +\infty,  \sigma_2(\omega_3, \omega_4) = 2, \sigma_3(\omega_3) = 2,  \ \sigma_3(\omega_1, \omega_2, \omega_4) = +\infty.$$
\end{lemma}
\begin{proof}
 By the definitions of $Z$ and $\widetilde{Z}$ in (\ref{eq:crucialZandZtilde}), we calculate that
 \begin{eqnarray*}
  Z_0 &=& P(\tau >0 ) = 1, \ Z_1 = P(\tau >1 | {\cal F}_1)   = \mathds{1}_{\{\omega_1, \omega_2 \}}  + \mathds{1}_{\{\omega_3, \omega_4\}} (1- p), \ Z_2 = 0\\
   \widetilde{Z}_0 &=& P(\tau \geq 0 ) = 1, \ \widetilde{Z}_1 = P(\tau \geq 1 | {\cal F}_1)   = 1, \  \widetilde{Z}_2 = P(\tau \geq 2 | {\cal F}_2) = \mathds{1}_{\{\omega_1, \omega_2, \omega_4\}}.
 \end{eqnarray*}
The calculation for $A$ and $m$ is analogical.  The assertions (b), (c) and (d) are straightforward to check. We omit it here.
\end{proof}

\begin{theorem}
Under the current settings, the following properties hold:\\
$\rm (a)$  In the market  $S^\tau$, both public traders with information $\mathbb{F}$ and insiders with information $\mathbb{G}$ can make  arbitrage. \\
$\rm (b)$ In the market $S$, only insiders with information $\mathbb{G}$ can make arbitrage. \\
$\rm (c)$ The conditions (\ref{eq:crucialassumptionbeforetau}) and (\ref{eq:crucialassumptionhonest}) failed. Indeed
 \begin{eqnarray}\label{eq:exmfail}
 \{\Delta S_2 \neq 0\}\cap  \{\widetilde{Z}_2 = 0\} \cap \{Z_1>0\} = \{\omega_3\},   \ \mbox{and}  \  \{\Delta S_2 \neq 0\}\cap\{\widetilde{Z}_2 = 1\} \cap \{Z_1<1\} = \{\omega_4\}.
 \end{eqnarray}
\end{theorem}

\begin{proof}
 In the market $S^\tau$, apparently, by taking $H_1 = 0, \ H_2(\{\omega_1, \omega_2\}) =0, \  H_2(\{\omega_3, \omega_4\}) = - 1$,   public traders and insiders can make an arbitrage. While in the market $S$, by taking $H_1^\mathbb{G} = 0, \ H_2^\mathbb{G}(\{\omega_1, \omega_2\}) =0, \  H_2^\mathbb{G}(\{\omega_3\}) = 1, H_2^\mathbb{G}(\{\omega_4\}) = -1$, only   insider traders   can make an arbitrage since the strategies $(H^\mathbb{G}_n)_{1\leq n \leq 2}$ are only $\mathbb{G}$-predictable. The condition (\ref{eq:exmfail}) is straightforward to verify.
\end{proof}

\begin{lemma}\label{lem:exampleGmart}
 For any $\mathbb{F}$-martingale  $M$, the following process  $M^\mathbb{G}$ is a $\mathbb{G}$-martingale.
 \begin{eqnarray*}
  M_0^\mathbb{G} = M_0, \ M_1^\mathbb{G}= M_1,   \  M_2^\mathbb{G} =  M_{2} -  p\mathds{1}_{\{\omega_4\}} \left( \Delta M_2(\omega_4) - \Delta M_2(\omega_3)  \right) + (1-p) \mathds{1}_{\{\omega_3\}}  \left( \Delta M_2(\omega_4) - \Delta M_2(\omega_3)  \right).
 \end{eqnarray*}
 \end{lemma}
\begin{proof}
Notice that $\Delta m_1 = 0$ and $\Delta m_2 = p \mathds{I}_{\{\omega_3, \omega_4\}}- \mathds{I}_{\{\omega_3 \}}$. We calculate that
\begin{eqnarray*}
 \frac{1}{Z_1}\mathds{1}_{\{2=\tau\}} E\left[\Delta M_2 \Delta m_2 \mid {\cal F}_1\right] &=& \frac{1}{1 - p \mathds{1}_{\{\omega_3, \omega_4\}}}\mathds{1}_{\{\omega_1, \omega_2, \omega_4\}}E\left[\Delta M_2\left( p \mathds{1}_{\{\omega_3, \omega_4\}}- \mathds{1}_{\{\omega_3 \}} \right)\mid {\cal F}_1\right] \\
 &=& \frac{1}{1 - p \mathds{1}_{\{\omega_3, \omega_4\}}}\mathds{1}_{\{\omega_4\}}E\left[\Delta M_2\left( p \mathds{1}_{\{\omega_3, \omega_4\}}- \mathds{1}_{\{\omega_3 \}} \right)\mid \{\omega_3, \omega_4\}\right]  \\
 &=& p\mathds{1}_{\{\omega_4\}} \left( \Delta M_2(\omega_4) - \Delta M_2(\omega_3)  \right),\\
 \frac{1}{1-Z_1}\mathds{1}_{\{1=\tau\}} E\left[\Delta M_2 \Delta m_2 \mid {\cal F}_1\right]  &=&p\mathds{1}_{\{\omega_3\}}E\left[\Delta M_2\left( p \mathds{1}_{\{\omega_3, \omega_4\}}- \mathds{1}_{\{\omega_3 \}} \right)\mid {\cal F}_1\right] \\
 &=& (1-p) \mathds{1}_{\{\omega_3\}}  \left( \Delta M_2(\omega_4) - \Delta M_2(\omega_3)  \right).
\end{eqnarray*}
The combination of  Theorem \ref{theo:GFmartau} and Theorem \ref{theo:GFmaraftertau} completes the proof of the lemma.
\end{proof}

\begin{corollary}
 The following process ${S^{\mathbb G}}$ is a $\mathbb{G}$-martingale:
\begin{eqnarray*}
  {S_0^{\mathbb G}}&=& S_0, \ {S_1^{\mathbb G}} (\{\omega_1, \omega_2\}) = u  {S}_0, \ \ {S_1^{\mathbb G}} (\{\omega_3, \omega_4\}) = d  {S}_0, \\
 {S_2^{\mathbb G}} (\{\omega_1 \}) &=& u^2  {S}_0,\  {S_2 ^{\mathbb G}}(\{\omega_2 \}) = u d  {S}_0, \ {S_2^{\mathbb G}} (\{\omega_3 \}) = d  {S}_0, {S_2^{\mathbb G}} (\{\omega_4 \}) = d  {S}_0.
\end{eqnarray*}
\end{corollary}
\begin{proof}
By the   Lemma  \ref{lem:exampleGmart}, we have
\begin{eqnarray*}
{S_2^{\mathbb G}} (\{\omega_3\}) &=& S_2(\{\omega_3\}) + (1-p)(d^2 - ud)S_0=  dS_0, \\
   {S_2^{\mathbb G}} (\{\omega_4 \}) &=&   {S}_2(\{\omega_4 \})- p (d^2 - ud)S_0  = dS_0,
\end{eqnarray*}
 where we used the equality $p u + (1-p)d =1$.
\end{proof}


{\bf Example \ref{example:2} (continued).}
We are now in the same setting as  Example \ref{example:2}. Recall the  random time
 \begin{equation}
   \tau_1 = \left\{
 \begin{array}{cc}
  1, & \mbox{ on } \  \{\omega_3\}\\
  2, & \mbox{ otherwise}
 \end{array}
 \right.
 \end{equation}
  The progressive enlargement filtration $\mathbb{G} = ({\cal G}_n)_{0\leq n \leq 2}$ is given by
\begin{eqnarray*}
{\cal G}_0 &=& \{\emptyset, \Omega\}, \ {\cal G}_1 = \{ \emptyset, \Omega, \{\omega_1, \omega_2\}, \{\omega_3\}, \{ \omega_4\}\}, \mbox{and }{\cal G}_2 =  \sigma(\{\emptyset, \Omega, \{\omega_1 \}, \{\omega_2 \}, \{\omega_3 \}, \{\omega_4 \} \}).
\end{eqnarray*}
\begin{lemma}
 For the above settings, we have
 \begin{eqnarray*}
  A_0 &=& 0, \ A_1 = \lambda \mathds{1}_{\{\omega_3, \omega_4\}},  \ A_2 = \lambda \mathds{1}_{\{\omega_3,  \omega_4\}} + \mathds{1}_{\{\omega_1, \omega_2, \omega_4\}}.  \\
  m_0 &=& 1, \ m_1 = 1, \ m_2 = \lambda \mathds{1}_{\{\omega_3,  \omega_4\}} + \mathds{1}_{\{\omega_1, \omega_2, \omega_4\}}.  \\
  Z_0 &=& 1, \  Z_1= 1 - \lambda  \mathds{1}_{\{\omega_3, \omega_4\}}, \ Z_2 = 0. \\
  \widetilde{Z}_0 &=& 1, \ \widetilde{Z}_1 = 1, \  \widetilde{Z}_2 = \mathds{1}_{\{\omega_1, \omega_2, \omega_4\}}.
 \end{eqnarray*}
 As a consequence,
 \begin{eqnarray*}
  \widetilde{Z}_{\tau_1} = 1, \ \ \mbox{and } \ \tau_1 \mbox{ is an honest time}.
 \end{eqnarray*}
\end{lemma}
\begin{proof}
 The calculations follow the same schedule as that of Lemma \ref{lem:exam1--}.
\end{proof}

\begin{theorem}
 The process $S$ stays as a $\mathbb{G}$-martingale. Therefore there is no arbitrage opportunity in the market $S^{\tau_1}$ and $S -S^{\tau_1}$; meanwhile \begin{eqnarray}\label{eq:exmhold}
 \{\Delta S_2 \neq 0\}\cap  \{\widetilde{Z}_2 = 0\} \cap \{Z_1>0\} = \emptyset,   \ \mbox{and}  \  \{\Delta S_2 \neq 0\}\cap\{\widetilde{Z}_2 = 1\} \cap \{Z_1<1\} = \emptyset.
 \end{eqnarray}
 \end{theorem}
\begin{proof}
It is easy to see that
\begin{eqnarray*}
  E\left[\Delta S_2 \Delta m_2 \mid {\cal F}_1\right] = E\left[\Delta S_2 \left( (\lambda - 1) \mathds{1}_{\{\omega_3\}} + \mathds{1}_{\{\omega_4\}} \right) \mid {\cal F}_1\right]= 0,
 \end{eqnarray*}
 where we used the fact that $\Delta S_2 \equiv 0$ on $\{\omega_3, \omega_4\}$. Therefore the process $S$ stays as a $\mathbb{G}$-martingale and there is no arbitrage.
\end{proof}

{\bf Acknowledgements:}   The research of Tahir Choulli  and Jun Deng is supported financially by the
Natural Sciences and Engineering Research Council of Canada,
through Grant G121210818.

\end{document}